\documentclass[a4paper,USenglish,cleveref,autoref,thm-restate]{lipics-v2021}

\usepackage{amsmath,amssymb}
\usepackage{stmaryrd}
\usepackage{etoolbox}
\usepackage{mathtools}
\usepackage{mathpartir}
\usepackage{upgreek}
\usepackage{prooftree}
\usepackage{xspace}

\nolinenumbers
\DOIPrefix{} % remove for real submission

\newif\ifcomments
\commentstrue
\newcommand{\mktag}[1]{\mathsf{\color{teal}#1}}
\newcommand{\mkkeyword}[1]{\mathsf{\color{blue}#1}}
\newcommand{\mkfunction}[1]{\mathsf{#1}}

\newcommand{\parens}[1]{(#1)}
\newcommand{\braces}[1]{\{#1\}}
\newcommand{\bracks}[1]{[#1]}
\newcommand{\angles}[1]{\langle#1\rangle}

\newcommand{\set}[1]{\braces{#1}}
\newcommand{\seqof}[1]{\overline{#1}}

\newcommand{\rulename}[1]{\textnormal{\textsc{\small[#1]}}}

\newcommand{\defrule}[2][]{\hypertarget{rule:\ifblank{#1}{#2}{#1}}{\rulename{#2}}}
\newcommand{\refrule}[2][]{\rulename{#2}\xspace}

\newcommand{\Nat}{\mathbb{N}}
\newcommand{\CSLL}{\textsf{CSLL}\xspace}
\newcommand{\CoreCSLL}{\texorpdfstring{$\textsf{\upshape CSLL}^\infty$}{coreCSLL}\xspace}
\newcommand{\isdet}{\textsf{\upshape det}}
\newcommand{\DetCSLL}{$\textsf{\upshape CSLL}^\infty_{\isdet}$\xspace}
\newcommand{\muMALL}{\texorpdfstring{$\mu\textsf{\upshape MALL}^\infty$}{muMALL}\xspace}

\newcommand{\SILLS}{$\textsf{SILL}_{\textsf{S}}$\xspace}
\newcommand{\HCPND}{$\textsf{HCP}_{\textsf{ND}}$\xspace}

\newcommand\eoe{\lipicsEnd}

\newcommand\A{\mathsf{A}}
\newcommand\B{\mathsf{B}}

\newcommand{\VarSet}{\mathcal{V}}
\newcommand{\ProcessSet}{\mathcal{P}}

\newcommand{\InTag}{\mktag{in}}
\newcommand{\LeftTag}{\InTag_1}
\newcommand{\RightTag}{\InTag_2}

\newcommand{\x}{x}
\newcommand{\y}{y}
\newcommand{\z}{z}
\renewcommand{\u}{u}
\renewcommand{\v}{v}

\newcommand{\parop}{\mathbin|}

\newcommand{\Let}[3]{#1\ifblank{#2}{}{\parens{#2}}\triangleq#3}
\newcommand{\Call}[2]{#1\ifblank{#2}{}{\angles{#2}}}
\newcommand{\Link}[2]{#1\leftrightarrow#2}
\newcommand{\Close}[1]{\mkkeyword{close}\,#1}
\newcommand{\Wait}[1]{\mkkeyword{wait}\,#1}
\newcommand{\Fail}[1]{\mkkeyword{fail}\,#1}
\newcommand{\Select}[2]{#2\,#1}
\newcommand{\Left}[1]{\Select{#1}{\InTag_1}}
\newcommand{\Right}[1]{\Select{#1}{\InTag_2}}
\newcommand{\CaseX}[3]{\mkkeyword{case}\,{#1}\braces{#2,#3}}
\newcommand{\Case}[3]{\CaseX{#1}{#2}{#3}}
\newcommand{\Fork}[4]{#1\bracks{#2}\parens{#3\parop#4}}
\newcommand{\Send}[2]{#1\angles{#2}}

\renewcommand{\Join}[2]{#1\parens{#2}}
\newcommand{\Cut}[3]{\parens{#1}\parens{#2\parop#3}}

\newcommand{\EmptyPool}[1]{\tclient#1\bracks{}}
\newcommand{\Cons}[3]{\Client{#1}{#2}.#3 \mathrel{::}}
\newcommand{\Nil}[1]{\tclient#1\bracks{}}
\newcommand{\Client}[2]{\tclient#1\bracks{#2}}
\newcommand{\Server}[4]{\tserver#1\parens{#2}\braces{#3,#4}}

\newcommand{\Hole}{\bracks\,}
\newcommand{\RC}{\mathcal{C}}
\newcommand{\RD}{\mathcal{D}}

\newcommand{\AddressSet}{\mathcal{A}}
\newcommand{\address}{\addressA}
\newcommand{\addressA}{\alpha}
\newcommand{\addressB}{\beta}

\newcommand{\Type}{\TypeT}
\newcommand{\TypeT}{T}
\newcommand{\TypeS}{S}

\newcommand{\Formulas}{\Phi}
\newcommand{\Formula}{\FormulaF}
\newcommand{\FormulaF}{\varphi}
\newcommand{\FormulaG}{\psi}

\newcommand{\X}{X}
\newcommand{\Y}{Y}

\newcommand{\mkformula}[1]{#1}

\newcommand{\Bot}{\mkformula\bot}
\newcommand{\Top}{\mkformula\top}
\newcommand{\One}{\mkformula{\mathbf{1}}}
\newcommand{\Zero}{\mkformula{\mathbf{0}}}
\newcommand{\choice}{\mathbin{\mkformula\oplus}}
\newcommand{\branch}{\mathbin{\mkformula\binampersand}}
\newcommand{\tfork}{\mathbin{\mkformula\otimes}}
\newcommand{\tjoin}{\mathbin{\mkformula\bindnasrepma}}
\newcommand{\tserver}{\text{\upshape!`}}
\newcommand{\tclient}{\text{\upshape?`}}
\newcommand{\tmu}{\mkformula\mu}
\newcommand{\tnu}{\mkformula\nu}

\newcommand{\Sequent}{\SequentS}
\newcommand{\SequentS}{\Sigma}
\newcommand{\SequentT}{\Theta}

\newcommand{\Context}{\ContextC}
\newcommand{\ContextC}{\Upgamma}
\newcommand{\ContextD}{\Updelta}

\newcommand{\wtp}[3][]{#2 \vdash\ifblank{#1}{}{{\color{red}XXX}} #3}

\newcommand{\CallRule}{call}

\newcommand{\CutRule}{cut}
\newcommand{\FailRule}{$\Top$}
\newcommand{\CloseRule}{$\One$}
\newcommand{\WaitRule}{$\Bot$}
\newcommand{\ForkRule}{$\tfork$}
\newcommand{\JoinRule}{$\tjoin$}
\newcommand{\SelectRule}{$\choice$}
\newcommand{\CaseRule}{$\branch$}
\newcommand{\RecRule}{$\tmu$}
\newcommand{\CorecRule}{$\tnu$}

\newcommand{\ServerRule}{server}

\newcommand{\ConsRule}{client}
\newcommand{\NilRule}{done}

\newcommand{\pcong}{\preccurlyeq}
\newcommand{\red}{\rightarrow}
\newcommand{\reds}{\Rightarrow}
\newcommand{\dred}{\red_{\isdet}}
\newcommand{\dreds}{\reds_{\isdet}}

\newcommand{\wred}{\Rightarrow}
\newcommand{\nred}{\arrownot\red}
\newcommand{\ndred}{\arrownot\dred}
\newcommand{\eqdef}{\stackrel{\smash{\textsf{\upshape\tiny def}}}=}

\newcommand{\prefix}{\sqsubseteq}
\newcommand{\subf}{\preceq}
\newcommand{\tred}{\leadsto}

\newcommand{\fn}[1]{\mkfunction{fn}\parens{#1}}
\newcommand{\bn}[1]{\mkfunction{bn}\parens{#1}}
\newcommand{\subst}[2]{\braces{#1/#2}}
\newcommand{\dual}[1]{#1^{\bot}}
\newcommand{\dom}[1]{\mkfunction{dom}\parens{#1}}
\newcommand{\InfOften}[1]{\mkfunction{inf}\parens{#1}}
\newcommand{\minf}{\mkfunction{min}\,}
\newcommand{\strip}[1]{\overline{#1}}

\newcommand{\threads}[1]{\mathsf{guards}\parens{#1}}
\newcommand{\channels}[1]{\mathsf{channels}\parens{#1}}

\theoremstyle{remark}

\title{On the Fair Termination of Client-Server Sessions}

\author{Luca Padovani}{Universit\`a di Camerino, Italy}{luca.padovani@unicam.it}{https://orcid.org/0000-0001-9097-1297}{}

\authorrunning{L. Padovani}
\Copyright{Luca Padovani}



\category{}

\relatedversion{}

\begin{document}

\maketitle

\begin{abstract}
    Client-server sessions are based on a variation of the traditional
    interpretation of linear logic propositions as session types in which
    non-linear channels (those regulating the interaction between a pool of
    clients and a single server) are typed by \emph{coexponentials} instead of
    the usual exponentials.
    Coexponentials enable the modeling of racing interactions, whereby clients
    compete to interact with a single server whose internal state (and thus the
    offered service) may change as the server processes requests sequentially. 
    In this work we present a fair termination result for \CoreCSLL, a core
    calculus of client-server sessions. We design a type system such that every
    well-typed term corresponds to a valid derivation in \muMALL, the infinitary
    proof theory of linear logic with least and greatest fixed points. We then
    establish a correspondence between reductions in the calculus and principal
    reductions in \muMALL. Fair termination in \CoreCSLL follows from cut
    elimination in \muMALL.
\end{abstract}

\section{Introduction}
\label{sec:introduction}

Session types~\cite{Honda93,HondaVasconcelosKubo98,HuttelEtAl16} are
descriptions of communication protocols enabling the static enforcement of a
variety of safety and liveness properties, including the fact that communication
channels are used according to their protocol (\emph{fidelity}), that processes
do not get stuck (\emph{deadlock freedom}), that pending communications are
eventually completed (\emph{livelock freedom}), that sessions eventually end
(\emph{termination}).
It is possible to trace a close correspondence between session types and
propositions of linear logic, and between the typing rules of a session type
system and the proof rules of linear
logic~\cite{Wadler14,CairesPfenningToninho16,LindleyMorris16}. This
correspondence provides session type theories with a solid logical foundation
and enables the application of known results concerning linear logic proofs into
the domain of communicating protocols. One notable example is \emph{cut
elimination}: the fact that every linear logic proof can be reduced to a form
that does not make use of the cut rule means that the process described by the
proof can be reduced to a form in which no pending communication is present,
provided that there is a good correspondence between cut reductions in proofs
and reductions in processes.

The development of session type systems based on linear logic also poses some
challenges with respect to their ability to cope with ``real-world'' scenarios.
An example, which is the focus of this work, is the modeling of the interactions
between a \emph{pool of clients} and a \emph{single server}. By definition, a
server is a process that can handle an unbounded number of requests made by
clients. In a session type system based on linear logic, it is natural to
associate the channel from which a server accepts client requests with a type of
the form $!\Type$, indicating the unlimited availability of a service with type
$\Type$. In fact, it was observed early on~\cite{GirardLafont87} that the
meaning of the ``of course'' modality $!\Type$ could be informally expressed by
the equation
\[
    !\Type \cong \One \branch \Type \branch (!\Type \tfork !\Type)
\]
which could be read ``as many copies of $\Type$ as the clients require''.
While appealing from a theoretical point of view, the association between the
concept of server and the ``of course'' modality is both unrealistic and
imprecise.
First of all, it models the ``unlimited'' availability of the server by means of
unlimited \emph{parallel} copies of the server, each copy dealing with a single
request, rather than by a \emph{single} process that is capable of handling an
unlimited number of requests \emph{sequentially}.
Second, it fails to capture the fact that each connection between a client and
the server may alter the server's internal state, in such a way that different
connections may potentially affect each other.

These considerations have led Qian et al.~\cite{QianKavvosBirkedal21} to develop
\CSLL (for ``Client-Server Linear Logic''), a session type system based on
linear logic which includes the \emph{coexponential} modalities $\tclient\Type$
and $\tserver\Type$ whose meaning can be (informally) expressed by the equations
\begin{equation}
    \label{eq:coexp}
    \tclient\Type \cong \One \choice \Type \choice (\tclient\Type \tfork \tclient\Type)
    \text{\qquad and\qquad}
    \tserver\Type \cong \Bot \branch \Type \branch (\tserver\Type \tjoin \tserver\Type)
\end{equation}
according to which a server that behaves as $\tserver\Type$ offers $\Type$ as
many times as necessary to satisfy all client requests, but it does so
sequentially and in some (unspecified) order.
%
% \CSLL turns out to be quite expressive, as witnessed by the fact that it is
% capable of modeling a Compare-And-Set service whose internal state is affected
% by (and affects) the connecting clients.
%
Qian et al.~\cite{QianKavvosBirkedal21} show that well-typed \CSLL processes are
deadlock free, but they leave a proof of termination to future work conjecturing
that it could be quite involved. A proof of this property is valuable since
termination (combined with deadlock freedom) implies livelock freedom.

In this paper we attack the problem of establishing a termination result for
\CSLL. Instead of providing an \emph{ad hoc} proof, we attempt to reduce the
termination problem for \CSLL to the cut elimination property of a known logical
system. To this aim, we propose a variation of \CSLL called \CoreCSLL that is in
close relationship with
\muMALL~\cite{BaeldeDoumaneSaurin16,Doumane17,BaeldeEtAl22}, the infinitary
proof theory of multiplicative-additive linear logic with least and greatest
fixed points. The basic idea is to encode the coexponentials in \CoreCSLL's type
system as fixed points in \muMALL following their expected meaning
(\cref{eq:coexp}). At this point, the cut elimination property of \muMALL should
allow us to deduce that well-typed \CoreCSLL processes do not admit infinite
reduction sequences.
As it turns out, we are unable to follow this plan of action in full. The
problem is that some reductions in \CoreCSLL do not correspond to cut reduction
steps in \muMALL. More specifically, even though clients are queued into client
pools, they should be able to reduce in any order, independently of their
position in the queue. This independent reduction of the clients in the same
pool is not matched by the sequence of cut reduction steps that are performed in
the cut elimination proof of \muMALL.
Still, the cut elimination property of \muMALL allows us to prove a useful
result, namely that every well-typed \CoreCSLL process is \emph{fairly
terminating}. Fair termination~\cite{GrumbergFrancezKatz84,Francez86} is weaker
than termination since it does not rule out the existence of infinite reduction
sequences. However, it guarantees that every \emph{fair} and maximal reduction
sequence of a well-typed \CoreCSLL process is finite, under a suitable fairness
assumption. In particular, fair termination is strong enough (when combined with
deadlock freedom) to guarantee livelock freedom.

The adoption of \muMALL as logical foundation for \CoreCSLL has another
advantage. In the original presentation of \CSLL~\cite{QianKavvosBirkedal21} the
process calculus is equipped with an unconventional operational semantics
whereby reductions can occur underneath prefixes and prefixes may be moved
around crossing restrictions, parallel compositions and other (unrelated)
prefixes. This semantics is justified to keep the process reduction rules and
the cut reduction rules sufficiently aligned, so that the cut elimination
property in the logic can be reflected to some valuable property in the
calculus, such as deadlock freedom. In contrast, \CoreCSLL features an entirely
conventional reduction semantics. We can afford to do so because \muMALL is an
\emph{infinitary} proof system in which the cut elimination property is proved
bottom-up by \emph{reducing outermost cuts first}. This reduction strategy
matches the ordinary reduction semantics of any process calculus in which
reductions happen at the outermost levels of processes. In the end, since the
reduction semantics of \CoreCSLL is stricter than that of \CSLL, the deadlock
freedom and the fair termination results we prove for \CoreCSLL are somewhat
stronger than their counterparts in the context of \CSLL.

\subparagraph*{Structure of the paper.}
\cref{sec:language} describes syntax and semantics of \CoreCSLL and defines the
notion of fairly terminating process.
We develop the type system for \CoreCSLL in \cref{sec:types}.
In \cref{sec:mumall} we recall the key elements of \muMALL, before addressing
the proof that well-typed \CoreCSLL processes fairly terminate in
\cref{sec:termination}.
\cref{sec:example} revisits an example of non-deterministic server given by Qian
et al.~\cite{QianKavvosBirkedal21} in our setting.
We summarize our results and further compare \CoreCSLL with
\CSLL~\cite{QianKavvosBirkedal21} and other related work in
\cref{sec:conclusion}.
Some proofs and definitions have been moved into \cref{sec:extra-types}.

\newcommand\Buyer{\textit{Buyer}}
\newcommand\Seller{\textit{Seller}}

\section{Syntax and Semantics of \CoreCSLL}
\label{sec:language}

\begin{table}
    \caption{\label{tab:syntax} Syntax of \CoreCSLL.}
    \centering
    \begin{math}
        \displaystyle
        \begin{array}[t]{@{}r@{~}c@{~}ll@{}}
            P, Q
            & ::= & \Call\A{\seqof\x} & \text{invocation}
            \\
            & | & \Fail\x & \text{failure}
            \\
            & | & \Wait\x.P & \text{wait}
            \\
            & | & \Join\x\y.P & \text{input}
            \\
            & | & \Case\x{P}{Q} & \text{branch}
            \\
            & | & \Server\x\y{P}{Q} & \text{server}
        \end{array}
        ~
        \begin{array}[t]{@{}r@{~}c@{~}lll@{}}
            & | & \Cut\x{P}{Q} & \text{parallel composition}
            \\
            & | & \Nil\x & \text{empty pool}
            \\
            & | & \Close\x & \text{close}
            \\
            & | & \Fork\x\y{P}{Q} & \text{output}
            \\
            & | & \Select\x{\InTag_i}.P & \text{select} & i\in\set{1,2}
            \\
            & | & \Cons\x\y{P}{Q} & \text{client pool}
        \end{array}
    \end{math}
\end{table}

In this section we define syntax and semantics of \CoreCSLL, a
calculus of sessions in which servers handle client requests
sequentially.
The syntax of \CoreCSLL makes use of an infinite set $\VarSet$ of
\emph{channels} ranged over by $x$, $y$ and $z$ and a set $\ProcessSet$ of
\emph{process names} ranged over by $\A$, $\B$, and so on. In \CoreCSLL channels
are of two kinds (which will be distinguished by their type): \emph{session
channels} connect two communicating processes; \emph{shared channels} connect an
unbounded number of clients with a single server. The structure of terms is
given by the grammar in \cref{tab:syntax} and their meaning is informally
described below.
%
% The term $\Link\x\y$ represents a \emph{forwarder} that routes every message
% sent on one of the channels to the other one.
%
The term $\Cut\x{P}{Q}$ represents the parallel composition of $P$ and $Q$
connected by the restricted channel $x$, which can be either a session channel
or a shared channel.
The term $\Fail\x$ represents a process that signals a failure on
channel $x$.
The term $\Close\x$ models the closing of a session, whereas
$\Wait\x.P$ models a process that waits for $x$ to be closed and
then continues as $P$.
The term $\Fork\x\y{P}{Q}$ models a process that creates a new
channel $y$, sends $y$ over $x$, uses $y$ as specified by $P$ and
$x$ as specified by $Q$.
%
% Note that this operation corresponds to a \emph{bound output}.  The
% free output can be modeled as the term
% $\Send\x\y.P \eqdef \Fork\x\z{\Link\z\y}{P}$.
%
The term $\Join\x\y.P$ models a process that receives a channel $y$
from $x$ and then behaves as $P$.
The term $\Select\x{\InTag_i}.P$ models a process that sends the label
$\InTag_i$ over $x$ and then behaves as $P$. In this work we only consider two
labels $\InTag_1$ and $\InTag_2$, although it is common to allow for an
arbitrary set of atomic labels. Dually, the term $\Case\x{P_1}{P_2}$ models a
process that waits for a label $\InTag_i$ from $x$ and then behaves according to
$P_i$.
The term $\Nil\x$ models the empty pool of clients connecting with a server on
the shared channel $x$, whereas the term $\Cons\x\y{P}Q$ models a client pool
consisting of a client that connects with a server on channel $x$ and behaves as
$P$ and another client pool $Q$. Occasionally we write $\Client\x\y.P$ instead
of $\Cons\x\y{P}\Nil\x$.
The term $\Server\x\y{P}{Q}$ models a server that waits for
connections on the shared channel $x$. If a new connection $y$ is
established, the server continues as $P$. If no clients are left
connecting on $x$, the service on $x$ is terminated and the process
continues as $Q$.
Finally, a term $\Call\A{\seqof\x}$ represents the invocation of the process
named $\A$ with arguments $\seqof\x$. We assume that each process name is
associated with a unique global definition of the form $\Let\A{\seqof\x}{P}$.
The notation $\seqof{e}$ is used throughout the paper to represent possibly
empty sequences $e_1,\dots,e_n$ of various entities.

The notions of free and bound names are defined in the expected way. Note that
the output operations $\Fork\x\y{P}{Q}$ and $\Cons\x\y{P}Q$ bind $y$ in $P$ but
not in $Q$. We write $\fn{P}$ and $\bn{P}$ for the sets of free and bound names
in $P$, we identify processes up to renaming of bound channel names and we
require $\fn{P} = \set{\seqof\x}$ for each global definition
$\Let\A{\seqof\x}{P}$.

\begin{table}
  \caption{\label{tab:semantics}Structrual pre-congruence and reduction semantics of \CoreCSLL.}
  \[
    \begin{array}{@{}rll@{}}
      \defrule{s-par-comm} &
      \Cut\x{P}{Q} \pcong \Cut\x{Q}{P}
      \\
      \defrule{s-pool-comm} &
      \multicolumn{2}{l@{}}{
        \Cons\x\y{P}\Cons\u\v{Q}R \pcong \Cons\u\v{Q}\Cons\x\y{P}R
      }
      \\
      \defrule{s-par-assoc} &
      \Cut\x{P}{\Cut\y{Q}{R}} \pcong \Cut\y{\Cut\x{P}{Q}}{R} &
      x\in\fn{Q}\setminus\fn{R}, y\not\in\fn{P}
      \\
      \defrule{s-pool-par} &
      \Cons\x\y{P}\Cut\z{Q}{R} \pcong \Cut\z{\Cons\x\y{P}Q}{R} &
      x \in \fn{Q}, z \not\in \fn{\Client\x\y.P}
      \\
      \defrule{s-par-pool} &
      \Cut\z{\Cons\x\y{P}Q}{R} \pcong \Cons\x\y{P}\Cut\z{Q}{R} &
      z\not\in\fn{\Client\x\y.P}
      \\
      \defrule{s-call} &
      \Call\A{\seqof\x} \pcong P &
      \Let\A{\seqof\x}{P}
      \\
      \\
      \defrule{r-close} &
      \Cut\x{\Close\x}{\Wait\x.P} \red P
      \\
      \defrule{r-comm} &
      \multicolumn{2}{l@{}}{
        \Cut\x{\Fork\x\y{P}{Q}}{\Join\x\y.R} \red \Cut\y{P}{\Cut\x{Q}{R}}
      }
      \\
      \defrule{r-case} &
      \Cut\x{\Select\x{\InTag_i}.P}{\Case\x{Q_1}{Q_2}} \red \Cut\x{P}{Q_i}
      \\
      \defrule{r-done} &
      \Cut\x{\EmptyPool\x}{\Server\x\y{P}{Q}} \red Q
      \\
      \defrule{r-connect} &
      \multicolumn{2}{l@{}}{
        \Cut\x{\Cons\x\y{P}Q}{\Server\x\y{R_1}{R_2}} \red \Cut\y{P}{\Cut\x{Q}{R_1}}
      }
      \\
      \defrule{r-par} &
      \Cut\x{P}{R} \red \Cut\x{Q}{R} &
      P \red Q
      \\
      \defrule{r-pool} &
      \Cons\x\y{R}P \red \Cons\x\y{R}Q &
      P \red Q
      \\
      \defrule{r-struct} &
      P \red Q &
      P \pcong P' \red Q' \pcong Q
    \end{array}
  \]
\end{table}

The operational semantics of \CoreCSLL is given by a structural precongruence
relation $\pcong$ and a reduction relation $\red$, both defined in
\cref{tab:semantics} and described below.
Rules \refrule{s-par-comm} and \refrule{s-pool-comm} state the expected
commutativity of parallel and pool compositions. In particular,
\refrule{s-pool-comm} allows clients in the same queue to swap positions,
modeling the fact that the order in which they connect to the server is not
deterministic.
Rule \refrule{s-par-assoc} models the associativity of parallel composition. The
side conditions make sure that no channel is captured ($y\not\in\fn{P}$) or left
dangling ($x\not\in\fn{R}$) and that parallel processes remain connected
($x\in\fn{Q}$).
The rules \refrule{s-pool-par} and \refrule{s-par-pool} deal with the mixed
associativity between parallel and pool compositions. The side conditions ensure
that no bound name leaves its scope and that parallel processes remain
connected.
Finally, \refrule{s-call} unfolds a process invocation to its definition.

Concerning the reduction relation, rule \refrule{r-close} models the closing of
a session, rule \refrule{r-comm} models the exchange of a channel and
\refrule{r-case} that of a label.
Rule \refrule{r-connect} models the connection of a client with a server,
whereas \refrule{r-done} deals with the case in which there are no clients left.
Finally, \refrule{r-par} and \refrule{r-pool} close reductions under parallel
compositions and client pools whereas \refrule{r-struct} allows reductions up to
structural pre-congruence.

Hereafter we write $\wred$ for the reflexive, transitive closure of
$\red$, we write $P \red$ if $P \red Q$ for some $Q$ and $P \nred$
if not $P \red$.
Later on we will also use a restriction of \CoreCSLL dubbed \DetCSLL
whose reduction relation, denoted by $\dred$, is obtained by
removing the rules \refrule{s-pool-comm}, \refrule{s-pool-par},
\refrule{s-par-pool} and \refrule{r-pool} (all those with
``\textsc{pool}'' in their name) from $\red$.  In essence, \DetCSLL
is a more deterministic version of \CoreCSLL in which clients are
forced to connect and reduce in the order in which they appear in
client pools. Also, clients are no longer allowed to cross
restricted channels.

\newcommand{\Lock}{\textsf{Lock}}
\newcommand{\User}{\textsf{User}}

\begin{example}
    \label{ex:lock}
    We illustrate the features of \CoreCSLL by modeling a pool of clients that
    compete to access a shared resource, represented as a simple \emph{lock}.
    When one client manages to acquire the lock, meaning that it has gained
    access to the resource, it prevents other clients from accessing the
    resource until the resource is released.
    We model the lock with this definition:
    \[
      \Let\Lock{x,z}{\Server\x\y{\Wait\y.\Call\Lock{x,z}}{\Close\z}}
    \]

    The lock is a server waiting for connections on the shared channel $x$,
    whereas each user is a client of the lock connecting on $x$. When a
    connection is established, the server waits until the resource is released,
    which is signalled by the termination of the session $y$, and then makes
    itself available again to handle further requests.

    The following process models the concurrent access to the lock by two clients:
    \[
        \Cut\x{\Cons\x\u{\Close\u}\Cons\x\v{\Close\v}\Nil\x}{\Call\Lock{x,z}}
    \]

    The order in which requests are handled by $\Lock$ is non-deterministic
    because of \refrule{s-pool-comm}. In this oversimplified example the users
    are indistinguishable and so non-determinism does not prevent the system to
    be confluent. In \cref{sec:example} we will see a more interesting example
    in which confluence is lost.
    This kind of interaction is typeable in \CoreCSLL thanks to coexponentials,
    which enable the concurrent access to a shared resource.
    \eoe
\end{example}

We conclude this section by defining various termination properties
of interest.
A \emph{run} of $P$ is a (finite or infinite) sequence $(P_0,P_1,\dots)$ of
processes such that $P = P_0$ and $P_i \red P_{i+1}$ whenever $P_{i+1}$ is a
term in the sequence. A run is \emph{maximal} if it is infinite or if it is
finite and its last term (say $Q$) cannot reduce any further (that is, $Q
\nred$).
We say that $P$ is \emph{terminating} if every maximal run of $P$ is finite. We
say that $P$ is \emph{weakly terminating} if $P$ has a maximal finite run.
A run of $P$ is \emph{fair} if it contains finitely many weakly terminating
processes. We say that $P$ is \emph{fairly terminating} if every fair run of $P$
is finite.

A fundamental property of any fairness notion is the fact that every
finite run of a process should be extendable to a maximal fair
one. This property, called
\emph{feasibility}~\cite{AptFrancezKatz87} or \emph{machine
  closure}~\cite{Lamport00}, holds for our fairness notion and
follows immediately from the next proposition.

\begin{proposition}
  \label{prop:feasibility}
  Every process has at least one maximal fair run.
\end{proposition}
\begin{proof}
  For an arbitrary process $P$ there are two possibilities.
  If $P$ is weakly terminating, then there exists $Q$ such that $P \wred Q
  \nred$. From this sequence of reductions we obtain a maximal run of $P$ that
  is fair since it is finite.
  If $P$ is not weakly terminating, then $P \red$ and $P \wred Q$ implies that
  $Q$ is not weakly terminating. In this case we can build an infinite run of
  $P$ which is fair since it does not go through any weakly terminating process.
\end{proof}

The given notion of fair termination admits an alternative characterization that
does not refer to fair runs. This characterization provides us with the key
proof principle to show that well-typed \CoreCSLL processes fairly terminate
(\cref{sec:termination}).

\begin{theorem}
  \label{thm:fair-termination}
  $P$ is fairly terminating iff $P \wred Q$ implies that $Q$ is weakly
  terminating.
\end{theorem}
\begin{proof}
  ($\Leftarrow$) Suppose by contradiction that $(P_0,P_1,\dots)$ is an infinite
  fair run of $P$ and note that $P \wred P_i$ for every $i$. From the hypothesis
  we deduce that every $P_i$ is weakly terminating. Then the run contains
  infinitely many weakly terminating processes, which is absurd by definition of
  fair run.
  ($\Rightarrow$)
  Suppose that $P \wred Q$. Then there is a finite run of $P$ that ends in $Q$.
  By \cref{prop:feasibility} there is a maximal fair run of $Q$. By
  concatenating these two runs we obtain a maximal fair run of $P$ that contains
  $Q$. From the hypothesis we deduce that this run is finite. Since $Q$ occurs
  in this run, we conclude that $Q$ is weakly terminating.
\end{proof}

%%% Local Variables:
%%% mode: latex
%%% TeX-master: "main"
%%% End:

\section{Type System}
\label{sec:types}

In this section we develop the type system of \CoreCSLL.  Types are
defined thus:
\[
    \textbf{Type}
    \qquad
    \TypeT, \TypeS ::= \Bot \mid \One \mid \Top \mid \Zero \mid \TypeT \tjoin \TypeS \mid \TypeT \tfork \TypeS \mid \TypeT \branch \TypeS \mid \TypeT \choice \TypeS \mid \tserver\Type \mid \tclient\Type
\]

Types extend the usual constants and connectives of multiplicative-additive
linear logic with the coexponentials $\tserver\Type$ and $\tclient\Type$ and, in
the context of \CoreCSLL, they describe how channels are used by processes.
Positive types indicate output operations whereas negative types indicate input
operations. In particular: $\One$/$\Bot$ describe a session channel used for
sending/receiving a session termination signal; $\Zero$/$\Top$ describe a
session channel used for sending/receiving an impossible (empty) message;
$\TypeT\tfork\TypeS$/$\TypeT\tjoin\TypeS$ describe a session channel used for
sending/receiving a channel of type $\TypeT$ and then according to $\TypeS$;
$\Type_1\choice\Type_2$/$\Type_1\branch\Type_2$ describe a session channel used
for sending/receiving a label $\InTag_i$ and then according to $\Type_i$;
finally, $\tclient\Type$/$\tserver\Type$ describe a shared channel used for
sending/receiving a connection message establishing a session of type $\Type$.
Each type $\Type$ has a \emph{dual} $\dual\Type$ obtained in the
expected way.  For example, we have
$\dual{(\One \choice \Type)} = \Bot \branch \dual\Type$ and
$\dual{(\tserver\Type)} = \tclient\dual\Type$.

\begin{table}
    \caption{\label{tab:typing-rules} Typing rules for \CoreCSLL.}
    \begin{mathpar}
        \inferrule[\defrule\CallRule]{
            \wtp{P}{\seqof{x : \Type}}
        }{
            \wtp{\Call\A{\seqof\x}}{\seqof{x : \Type}}
        }
        ~\Let\A{\seqof\x}{P}
        \and
        \inferrule[\defrule\CutRule]{
            \wtp{P}{\ContextC, x : \Type}
            \\
            \wtp{Q}{\ContextD, x : \dual\Type}
        }{
            \wtp{\Cut\x{P}{Q}}{\ContextC, \ContextD}
        }
        \and
        \inferrule[\defrule\FailRule]{~}{
            \wtp{\Fail\x}{\Context, x : \Top}
        }
        \and
        \inferrule[\defrule\WaitRule]{
            \wtp\Context{P}
        }{
            \wtp{\Context, x : \Bot}{\Wait\x.P}
        }
        \and
        \inferrule[\defrule\CloseRule]{~}{
            \wtp{\Close\x}{x : \One}
        }
        \and
        \inferrule[\defrule\JoinRule]{
            \wtp{P}{\Context, y : \TypeT, x : \TypeS}
        }{
            \wtp{\Join\x\y.P}{\Context, x : \TypeT \tjoin \TypeS}
        }
        \and
        \inferrule[\defrule\ForkRule]{
            \wtp{P}{\ContextC, y : \TypeT}
            \\
            \wtp{Q}{\ContextD, x : \TypeS}
        }{
            \wtp{\Fork\x\y{P}{Q}}{\ContextC, \ContextD, x : \TypeT \tfork \TypeS}
        }
        \and
        \inferrule[\defrule\CaseRule]{
            \wtp{P}{\Context, x : \TypeT}
            \\
            \wtp{Q}{\Context, x : \TypeS}
        }{
            \wtp{\Case\x{P}{Q}}{\Context, x : \TypeT \branch \TypeS}
        }
        \and
        \inferrule[\defrule\SelectRule]{
            \wtp{P}{\Context, x : \Type_i}
        }{
            \wtp{\Select\x{\InTag_i}.P}{\Context, x : \Type_1 \choice \Type_2}
        }
        \and
        \inferrule[\defrule\ServerRule]{
            \wtp{P}{\Context, x : \tserver\Type, y : \Type}
            \\
            \wtp{Q}\Context
        }{
            \wtp{\Server\x\y{P}{Q}}{\Context, x : \tserver\Type}
        }
        \and
        \inferrule[\defrule\ConsRule]{
          \wtp{P}{\ContextC, y : \Type}
          \\
          \wtp{Q}{\ContextD, x : \tclient\Type}
        }{
            \wtp{\Cons\x\y{P}Q}{\ContextC, \ContextD, x : \tclient\Type}
        }
        \and
        \inferrule[\defrule\NilRule]{~}{
            \wtp{\EmptyPool\x}{x : \tclient\Type}
        }
    \end{mathpar}
\end{table}

The typing rules for \CoreCSLL are shown in
\cref{tab:typing-rules}. Typing judgments have the form
$\wtp{P}\Context$ and relate a process $P$ with a \emph{context}
$\Context$.
Contexts are finite maps from channel names to types written as
$\seqof{x:\Type}$. We let $\ContextC$ and $\ContextD$ range over
contexts, we write $\emptyset$ for the empty context, we write
$\dom\Context$ for the domain of $\Context$, namely for the set of
channel names for which there is an association in $\Context$, and
we write $\ContextC,\ContextD$ for the union of $\ContextC$ and
$\ContextD$ when $\dom\ContextC \cap \dom\ContextD = \emptyset$.

For the most part, the typing rules coincide with those of a
standard session type system based on linear
logic~\cite{Wadler14,LindleyMorris16}. In particular,
\refrule\CutRule, \refrule\FailRule, \refrule\WaitRule,
\refrule\CloseRule, \refrule\JoinRule, \refrule\ForkRule,
\refrule\CaseRule and \refrule\SelectRule relate the standard proof
rules of multiplicative-additive classical linear logic with the
corresponding forms of \CoreCSLL.
The rule \refrule\CallRule deals with process invocations
$\Call\A{\seqof\x}$ by unfolding the global definition of $\A$,
noted as side condition to the rule.
Rule \refrule\ServerRule deals with servers $\Server\x\y{P}{Q}$. The
continuation $P$, which is the actual handler of incoming
connections, must be well typed in a context enriched with the
channel $y$ resulting from the connection. Note that $x$ is still
present in the context and with the same type, meaning that $P$ must
\emph{also} be able to handle any further connection on the shared
channel $x$. The continuation $Q$, which models the behavior of the
server once no more clients are connecting on $x$, is not supposed
to use $x$ any longer.
Rule \refrule\ConsRule deals with non-empty client pools
$\Cons\x\y{P}Q$. The client $P$ is connecting with a server through
a shared channel $x$ and establishes a session $y$. The rest of the
pool $Q$ is using $x$ in the same way.
Rule \refrule\NilRule deals with the empty pool of clients
connecting on $x$.

The typing rules are interpreted coinductively. Therefore, a judgment
$\wtp{P}\Context$ is derivable if there is a possibly infinite typing derivation
for it. The need for infinite typing derivations stems from the fact that we
type process invocations by ``unfolding'' them to the process they represent, so
this unfolding may go on forever in the case of recursive processes.

\begin{example}
    \label{ex:lock-typing}
    Let us consider once again the process definitions in \cref{ex:lock}. We
    derive
    \[
        \begin{prooftree}
            \[
                \[
                    \[
                        \smash\vdots
                        \justifies
                        \wtp{\Call\Lock{x,z}}{x : \tserver\Bot, z : \One}
                        \using\refrule\CallRule
                    \]
                    \justifies
                    \wtp{\Wait\y.\Call\Lock{x,z}}{x : \tserver\Bot, y : \Bot, z : \One}
                    \using\refrule\WaitRule
                \]
                \[
                    \justifies
                    \wtp{\Close\z}{z : \One}
                    \using\refrule\CloseRule
                \]
                \justifies
                \wtp{\Server\x\y{\Wait\y.\Call\Lock{x,z}}{\Close\z}}{x : \tserver\Bot, z : \One}
                \using\refrule\ServerRule
            \]
            \justifies
            \wtp{\Call\Lock{x,z}}{x : \tserver\Bot, z : \One}
            \using\refrule\CallRule
        \end{prooftree}
    \]
    showing that $\Lock$ is well typed. Note that the typing derivation is
    infinite since $\Lock$ is a recursive process.
    We can now obtain the following typing derivation
    \[
        \begin{prooftree}
            \[
              \[
                \justifies
                \wtp{\Close\u}{u : \One}
                \using\refrule\CloseRule
                \]
                \[
                  \[
                    \justifies
                    \wtp{\Close\v}{v : \One}
                    \using\refrule\CloseRule
                  \]
                  \[
                    \justifies
                    \wtp{\Nil\x}{x : \tclient\One}
                    \using\refrule\NilRule
                  \]
                  \justifies
                  \wtp{
                    \Cons\x\v{\Close\v}
                    \Nil\x
                  }{
                    x : \tclient\One
                  }
                  \using\refrule\ConsRule
                \]
                \justifies
                \wtp{
                  \Cons\x\u{\Close\u}
                  \Cons\x\v{\Close\v}
                  \Nil\x
                }{
                    x : \tclient\One
                }
                \using\refrule\ConsRule
            \]
            \[
                \smash\vdots
                \justifies
                \wtp{\Call\Lock{x,z}}{x : \tserver\Bot, z : \One}
            \]
            \justifies
            \wtp{
                \Cut\x{
                  \Cons\x\u{\Close\u}
                  \Cons\x\v{\Close\v}
                  \Nil\x
                }{
                  \Call\Lock{x,z}
                }
            }{
                z : \One
            }
            \using\refrule\CutRule
        \end{prooftree}
    \]
    showing that the system as a whole is well typed.
    \eoe
\end{example}

Adopting an infinitary type system will make it easy to relate \CoreCSLL with
\muMALL (\cref{sec:termination}). However, we must be careful in that some
infinite typing derivations allow us to type processes that are not weakly
terminating, as illustrated in the next example.

\begin{example}[non-terminating process]
    \label{ex:omega}
    Consider the process $\Let\Omega{}{\Cut\x{\Close\x}{\Wait\x.\Omega}}$ which
    creates a session $x$, immediately closes it and then repeats the same
    behavior. Clearly, this process is not weakly terminating because it can
    only reduce thus:
    \[
        \Call\Omega{} \pcong
        \Cut\x{\Close\x}{\Wait\x.\Call\Omega{}} \red
        \Call\Omega{} \pcong
        \Cut\x{\Close\x}{\Wait\x.\Call\Omega{}} \red
        \cdots
    \]

    Nonetheless, we are able to find the following (infinite) typing
    derivation for $\Call\Omega{}$.
    \[
        \begin{prooftree}
            \[
                \[
                    \justifies
                    \wtp{
                        \Close\x
                    }{
                        x : \One
                    }
                    \using\refrule\CloseRule
                \]
                \[
                    \[
                        \smash\vdots\mathstrut
                        \justifies
                        \wtp{
                            \Call\Omega{}
                        }{
                            \emptyset
                        }
                        \using\refrule\CallRule
                    \]
                    \justifies
                    \wtp{
                        \Wait\x.\Call\Omega{}
                    }{
                        x : \Bot
                    }
                    \using\refrule\WaitRule
                \]
                \justifies
                \wtp{
                    \Cut\x{\Close\x}{\Wait\x.\Call\Omega{}}
                }{
                    \emptyset
                }
                \using\refrule\CutRule
            \]
            \justifies
            \wtp{
                \Call\Omega{}
            }{
                \emptyset
            }
            \using\refrule\CallRule
        \end{prooftree}
    \]

    Since we aim at ensuring fair termination for well-typed
    processes, we must consider this derivation as \emph{invalid}.
    \eoe
\end{example}

In order to rule out processes like $\Omega$ in \cref{ex:omega}, we identify a
class of valid typing derivations as follows.

\begin{definition}[valid typing derivation]
    \label{def:valid-typing}
    A typing derivation is \emph{valid} if every infinite branch in it goes
    through infinitely many applications of the rule \refrule\ServerRule
    concerning the same channel.
\end{definition}

This validity condition requires that every infinite branch of a typing
derivation describes the behavior of a server willing to accept an unbounded
number of connection requests.
If we look back at the infinite typing derivation for the $\Lock$ process in
\cref{ex:lock-typing}, we see that it is valid according to
\cref{def:valid-typing} since the only infinite branch in it goes through
infinitely many applications of the rule \refrule\ServerRule concerning the very
same shared channel $x$.
On the contrary, the typing derivation in \cref{ex:omega} is invalid
since the infinite branch in it does not go through any application
of \refrule\ServerRule.

The fact that every infinite branch must go through infinitely many
applications of \refrule\ServerRule \emph{concerning the very same
  shared channel} is a subtle point. Without the specification that
it is the \emph{same} shared channel to be found infinitely often,
it would be possible to obtain invalid typing derivations as
illustrated by the next example.

\begin{example}
    \label{ex:non-terminating-server}
    \newcommand{\OmegaServer}{\Omega\textsf{-Server}}
    Consider the definition
    \[
        \Let\OmegaServer\x{
            \Server\x\y{
                \Wait\y.\Call\OmegaServer\x
            }{
                \Cut\z{\EmptyPool\z}{\Call\OmegaServer\z}
            }
        }
    \]
    describing a server that waits for connections on the shared channel $x$.
    After each request, the server makes itself available again for handling
    more requests by the recursive invocation $\Call\OmegaServer\x$. Once all
    requests have been processed, the server creates a new shared channel on
    which an analogous server operates.
    Using the typing rules in \cref{tab:typing-rules} we are able to find the
    following typing derivation:
    \[
        \begin{prooftree}
            \[
                \[
                    \[
                        \smash\vdots\mathstrut
                        \justifies
                        \wtp{\Call\OmegaServer\x}{x : \tserver\Bot}
                        \using\refrule\CallRule
                    \]
                    \justifies
                    \wtp{
                        \Wait\y.\Call\OmegaServer\x
                    }{
                        x : \tserver\Bot,
                        y : \Bot
                    }
                    \using\refrule\WaitRule
                \]
                \[
                    \[
                        \justifies
                        \wtp{\EmptyPool\z}{z : \tclient\One}
                        \using\refrule\NilRule
                    \]
                    \[
                        \smash\vdots
                        \justifies
                        \wtp{\Call\OmegaServer\z}{z : \tserver\Bot}
                        \using\refrule\CallRule
                    \]
                    \justifies
                    \wtp{
                        \Cut\z{\EmptyPool\z}{\Call\OmegaServer\z}
                    }{
                        \emptyset
                    }
                    \using\refrule\CutRule
                \]
                \justifies
                \wtp{
                    \Server\x\y{
                        \Wait\y.\Call\OmegaServer\x
                    }{
                        \Cut\z{\EmptyPool\z}{\Call\OmegaServer\z}
                    }
                }{
                    x : \tserver\Bot
                }
                \using\refrule\ServerRule
            \]
            \justifies
            \wtp{\Call\OmegaServer\x}{x : \tserver\Bot}
            \using\refrule\CallRule
        \end{prooftree}
    \]

    Notice that the derivation bifurcates in correspondence of the application
    of \refrule\ServerRule and also that each sub-tree is infinite, since it
    contains an unfolding of the $\OmegaServer$ process.
    For this reason, the derivation contains (infinitely) many
    infinite branches, which are obtained by either ``going left''
    or ``going right'' each time \refrule\ServerRule is
    encountered. Each of these infinite branches goes through an
    application of \refrule\ServerRule infinitely many times, as
    requested by \cref{def:valid-typing}.
    Also, any such branch that ``goes right'' finitely many times
    eventually ends up going through infinitely many applications of
    \refrule\ServerRule that concern the same channel. In contrast,
    any branch that ``goes right'' infinitely many times keeps going
    through applications of \refrule\ServerRule concerning new
    shared channels created in correspondence of the application of
    \refrule\CutRule.
    In conclusion, this typing derivation is invalid and rightly so,
    or else the diverging process
    $\Cut\x{\EmptyPool\x}{\Call\OmegaServer\x}$ would be well typed
    in the empty context.
    \eoe
\end{example}

We conclude this section by stating two key properties of the type
system, starting from the fact that typing is preserved by
structural pre-congruence and reductions.

\begin{theorem}
    \label{thm:preservation}
    \newcommand{\rrel}{\mathrel{\mathcal{R}}}
    Let $P \rrel Q$ where ${\rrel}\in\set{{\pcong},{\red}}$. Then
    $\wtp{P}\Context$ implies $\wtp{Q}\Context$.
\end{theorem}

Also, processes that are well typed in a context of the form
$x : \One$ are deadlock free.

\begin{restatable}[deadlock freedom]{theorem}{thmdf}
    \label{thm:df}
    If $\wtp{P}{x : \One}$ then either $P \pcong \Close\x$ or $P \dred$.
\end{restatable}

Note that \cref{thm:df} uses $\dred$ instead of $\red$ in order to state that
$P$ is able to reduce if it is not (structurally pre-congruent to) $\Close\x$.
Recalling that ${\dred} \subseteq {\red}$, the deadlock freedom property ensured
by \cref{thm:df} is slightly stronger than one would normally expect. This
formulation will be necessary in \cref{sec:termination} when proving the
soundness of the type system. The proofs of \cref{thm:preservation,thm:df} can
be found in \cref{sec:extra-types}.

\newcommand{\ExpandLink}[1]{\textsf{Link}_{#1}}

\begin{example}[forwarder]
    \label{ex:forwarder}
    Most session calculi based on linear logic include a form $\Link\x\y$ whose
    typing rule $\wtp{\Link\x\y}{x : \Type, y : \dual\Type}$ corresponds to the
    axiom $\vdash \Type, \dual\Type$ of linear logic.
    The form $\Link\x\y$ is usually interpreted as a forwarder between the
    channels $x$ and $y$ and it is useful for example to model the output of a
    free channel $\Send\x\y.P$ as the term $\Fork\x\z{\Link\y\z}{P}$.
    In this example we show that there is no need to equip \CoreCSLL
    with a native form $\Link\x\y$ since its behavior can be encoded
    as a well-typed \CoreCSLL process. To this aim, we define a
    family $\ExpandLink\Type$ of process definitions by induction on
    $\Type$ as follows
    \begin{align*}
        \Let{\ExpandLink\Bot}{x,y&}{\Wait\x.\Close\y} \\
        \Let{\ExpandLink\Top}{x,y&}{\Fail\x} \\
        \Let{\ExpandLink{\TypeT \tjoin \TypeS}}{x,y&}{
            \Join\x\u.\Fork\y\v{
                \Call{\ExpandLink\TypeT}{u,v}
            }{
                \Call{\ExpandLink\TypeS}{x,y}
            }
        } \\
        \Let{\ExpandLink{\TypeT \branch \TypeS}}{x,y&}{
            \Case\x{\Left\y.\Call{\ExpandLink\TypeT}{x,y}}{\Right\y.\Call{\ExpandLink\TypeS}{x,y}}
        } \\
        \Let{\ExpandLink{\tserver\Type}}{x,y&}{
            \Server\x\u{
                \Cons\y\v{
                    \Call{\ExpandLink\TypeT}{u,v}
                }{
                    \Call{\ExpandLink{\tserver\Type}}{x,y}
                }
            }{
                \Nil\y
            }
        }
    \end{align*}
    with the addition of the definitions
    $\Let{\ExpandLink\Type}{x,y}{\Call{\ExpandLink{\dual\Type}}{y,x}}$ for the
    positive type constructors.
    It is easy to build a typing derivation for the judgment
    $\wtp{\Call{\ExpandLink\Type}{x,y}}{x : \Type, y : \dual\Type}$.
    Also, every infinite branch in such derivation eventually loops through an
    invocation of the form $\Call{\ExpandLink{\tserver\TypeS}}{u,v}$, which goes
    through an application of \refrule\ServerRule concerning the channel $u$.
    So, the derivation of $\wtp{\Call{\ExpandLink\Type}{x,y}}{x : \Type, y :
    \dual\Type}$ is valid and the process $\Call{\ExpandLink\Type}{x,y}$ is well
    typed.
    \eoe
\end{example}

%%% Local Variables:
%%% mode: latex
%%% TeX-master: "main"
%%% End:

\section{A quick recollection of \muMALL}
\label{sec:mumall}

In this section we recall the main elements of
\muMALL~\cite{Doumane17,BaeldeDoumaneSaurin16,BaeldeEtAl22}, the
infinitary proof system of the multiplicative additive fragment of
linear logic extended with least and greatest fixed points.
The syntax of \muMALL \emph{pre-formulas} makes use of an infinite
set of \emph{propositional variables} ranged over by $X$ and $Y$ and
is given by the grammar below:
\[
  \textbf{Pre-formula}
  \qquad
  \FormulaF, \FormulaG ::= \X \mid \Bot \mid \Top \mid \Zero \mid \One \mid \FormulaF \tjoin \FormulaG \mid \FormulaF \tfork \FormulaG \mid \FormulaF \branch \FormulaG \mid \FormulaF \choice \FormulaG \mid \tnu\X.\FormulaF \mid \tmu\X.\FormulaF
\]

The fixed point operators $\tmu$ and $\tnu$ are the binders of
propositional variables and the notions of free and bound variables
are defined accordingly. A \muMALL \emph{formula} is a closed
pre-formula. We write $\subst\Formula\X$ for the capture-avoiding
substitution of all free occurrences of $X$ with $\Formula$ and
$\dual\Formula$ for the \emph{dual} of $\Formula$, which is the
involution such that
\[
    \dual\X = X
    \qquad
    \dual{(\tmu\X.\Formula)} = \tnu\X.\dual\Formula
    \qquad
    \dual{(\tnu\X.\Formula)} = \tmu\X.\dual\Formula
\]
among the other expected equations. Postulating that $\dual\X = \X$
is not a problem since we will always dualize formulas, which do not
contain free propositional variables.

We write $\subf$ for the \emph{subformula ordering}, that is the
least partial order such that $\FormulaF \subf \FormulaG$ if
$\FormulaF$ is a subformula of $\FormulaG$. For example, if
$\FormulaF \eqdef \tmu\X.\tnu\Y.(X \choice Y)$ and
$\FormulaG \eqdef \tnu\Y.(\FormulaF \choice Y)$ we have
$\FormulaF \subf \FormulaG$ and
$\FormulaG \not\subf \FormulaF$.
When $\Formulas$ is a set of formulas, we write $\minf\Formulas$ for
its $\subf$-minimum formula if it is defined.
Occasionally we let $\star$ stand for an arbitrary binary connective
(one of $\choice$, $\tfork$, $\branch$, or $\tjoin$) and $\sigma$
stand for an arbitrary fixed point operator (either $\tmu$ or
$\tnu$).

In \muMALL it is important to distinguish among different
\emph{occurrences} of the same formula in a proof derivation. To
this aim, formulas are annotated with \emph{addresses}.
We assume an infinite set $\AddressSet$ of \emph{atomic addresses},
$\dual\AddressSet$ being the set of their duals such that
$\AddressSet \cap \dual\AddressSet = \emptyset$ and
$\dual{\dual\AddressSet{}} = \AddressSet$. We use $a$ and $b$ to
range over elements of $\AddressSet \cup \dual\AddressSet$.  An
\emph{address} is a string $aw$ where $w \in \set{i,l,r}^*$. The
dual of an address is defined as $\dual{(aw)} = \dual{a}w$.
We use $\addressA$ and $\addressB$ to range over addresses, we write
$\prefix$ for the prefix relation on addresses and we say that
$\addressA$ and $\addressB$ are \emph{disjoint} if
$\addressA \not\prefix \addressB$ and
$\addressB \not\prefix \addressA$.

A \emph{formula occurrence} (or simply occurrence) is a pair
$\Formula_\address$ made of a formula $\Formula$ and an address
$\address$. We use $F$ and $G$ to range over occurrences and we
extend to occurrences several operations defined on formulas. In
particular: we use logical connectives to compose occurrences so
that
$\FormulaF_{\address l} \mathbin\star \FormulaG_{\address r} \eqdef
(\FormulaF \mathbin\star \FormulaG)_\address$ and
$\sigma\X.\Formula_{\address i} \eqdef
(\sigma\X.\Formula)_\address$; the dual of an occurrence is obtained
by dualizing both its formula and its address, that is
$\dual{(\Formula_\address)} \eqdef \dual\Formula_{\dual\address}$;
occurrence substitution preserves the address in the type within
which the substitution occurs, but forgets the address of the
occurrence being substituted, that is
$\FormulaF_\addressA\subst{\FormulaG_\addressB}\X \eqdef
\FormulaF\subst\FormulaG\X_\addressA$.

We write $\strip{F}$ for the formula obtained by forgetting the
address of $F$.  Finally, we write $\tred$ for the least reflexive
relation on types such that $F_1 \star F_2 \tred F_i$ and
$\sigma\X.F \tred F\subst{\sigma\X.F}\X$.

\begin{table}
    \caption{\label{tab:mumall} Proof rules of \muMALL~\cite{BaeldeDoumaneSaurin16,Doumane17,BaeldeEtAl22}.}
    \begin{mathpar}
        % \inferrule[\refrule\LinkRule]{
        %     F \equiv G
        % }{
        %     \vdash F, \dualG
        % }
        % \and
        \inferrule[\refrule\CutRule]{
            \vdash \SequentS, F
            \\
            \vdash \SequentT, \dual{F}
        }{
            \vdash \SequentS, \SequentT
        }
        \and
        \inferrule[\refrule\FailRule]{~}{
            \vdash \Sequent, \Top
        }
        \and
        \inferrule[\refrule\WaitRule]{
            \vdash \Sequent
        }{
            \vdash \Sequent, \Bot
        }
        \and
        \inferrule[\refrule\CloseRule]{~}{
            \vdash \One
        }
        \and
        \inferrule[\refrule\JoinRule]{
            \vdash \Sequent, F, G
        }{
            \vdash \Sequent, F \tjoin G
        }
        \and
        \inferrule[\refrule\ForkRule]{
            \vdash \SequentS, F
            \\
            \vdash \SequentT, G
        }{
            \vdash \SequentS, \SequentT, F \tfork G
        }
        \and
        \inferrule[\refrule\CaseRule]{
            \vdash \Sequent, F
            \\
            \vdash \Sequent, G
        }{
            \vdash \Sequent, F \branch G
        }
        \and
        \inferrule[\refrule\SelectRule]{
            \vdash \Sequent, F_i
        }{
            \vdash \Sequent, F_1 \choice F_2
        }
        \and
        \inferrule[\refrule\CorecRule]{
            \vdash \Sequent, F\subst{\tnu\X.F}\X
        }{
            \vdash \Sequent, \tnu\X.F
        }
        \and
        \inferrule[\refrule\RecRule]{
            \vdash \Sequent, F\subst{\tmu\X.F}\X
        }{
            \vdash \Sequent, \tmu\X.F
        }
    \end{mathpar}
\end{table}

The proof rules of \muMALL are shown in \cref{tab:mumall}, where
$\SequentS$ and $\SequentT$ range over sets of occurrences written
as $F_1, \dots, F_n$. The rules allow us to derive \emph{sequents}
of the form $\vdash \Sequent$ and are standard except for
\refrule\CorecRule, which \emph{unfolds} a greatest fixed point just
like \refrule\RecRule does.
Being an infinitary proof system, \muMALL rules are meant to be
interpreted coinductively. That is, a sequent $\vdash \Sequent$ is
derivable if there exists an \emph{arbitrary} (finite or infinite)
proof derivation whose conclusion is $\vdash \Sequent$. Without a
\emph{validity condition} on derivations, such proof system is
notoriously unsound.
\muMALL's validity condition requires every infinite branch of a
derivation to be supported by the continuous unfolding of a greatest
fixed point. In order to formalize this condition, we start by
defining \emph{threads}, which are sequences of occurrences.

\begin{definition}[thread]
    \label{def:thread}
    A \emph{thread} of $F$ is a (finite or infinite) sequence of
    occurrences $(F_0,F_1,\dots)$ such that $F_0 = F$ and
    $F_i \tred F_{i+1}$ whenever $i+1$ is a valid index of the
    sequence.
\end{definition}

Hereafter we use $t$ to range over threads. For example, if we
consider $\Formula \eqdef \tmu\X.(X \choice \One)$, we have that
$t \eqdef
(\Formula_a,(\Formula\choice\One)_{ai},\Formula_{ail},\dots)$ is an
infinite thread of $\Formula_a$.
%
% A thread is \emph{stationary} if it has an infinite suffix of equal types. The
% above thread $t$ is not stationary.

Among all threads, we are interested in finding those in which a
$\tnu$-formula is unfolded infinitely often. These threads, called
$\tnu$-threads, are precisely defined thus:

\begin{definition}[$\tnu$-thread]
    \label{def:nu-thread}
    Let $t = (F_0,F_1,\dots)$ be an infinite thread, let $\strip{t}$ be the
    corresponding sequence $(\strip{F_0},\strip{F_1},\dots)$ of formulas and let
    $\InfOften{t}$ be the set of elements of $\strip{t}$ that occur infinitely
    often in $\strip{t}$. We say that $t$ is a \emph{$\tnu$-thread} if
    $\minf\InfOften{t}$ is defined and is a $\tnu$-formula.
\end{definition}

If we consider the infinite thread $t$ above, we have
$\InfOften{t} = \set{\Formula, \Formula \choice \One}$ and
$\minf\InfOften{t} = \Formula$, so $t$ is \emph{not} a $\tnu$-thread
because $\Formula$ is not a $\tnu$-formula.
Consider instead $\FormulaF \eqdef \tnu\X.\tmu\Y.(X \choice Y)$ and
$\FormulaG \eqdef \tmu\Y.(\FormulaF \choice Y)$ and observe that
$\FormulaG$ is the ``unfolding'' of $\FormulaF$. Now
$t_1 \eqdef (\FormulaF_a, \FormulaG_{ai}, (\FormulaF \choice
\FormulaG)_{aii}, \FormulaF_{aiil}, \dots)$ is a thread of
$\FormulaF_a$ such that
$\InfOften{t_1} = \set{\FormulaF, \FormulaG, \FormulaF \choice
  \FormulaG}$ and we have $\minf\InfOften{t_1} = \FormulaF$ because
$\FormulaF \subf \FormulaG$, so $t_1$ is a $\tnu$-thread.
If, on the other hand, we consider the thread
$t_2 \eqdef (\FormulaF_a, \FormulaG_{ai}, (\FormulaF \choice
\FormulaG)_{aii}, \FormulaG_{aiir}, (\FormulaF \choice
\FormulaG)_{aiiri}, \dots)$ such that
$\InfOften{t_2} = \set{\FormulaG, \FormulaF \choice \FormulaG}$ we
have $\minf\InfOften{t_2} = \FormulaG$ because
$\FormulaG \subf \FormulaF \choice \FormulaG$, so $t_2$ is not a
$\tnu$-thread.
Intuitively, the $\subf$-minimum formula among those that occur
infinitely often in a thread is the outermost fixed point operator
that is being unfolded infinitely often. It is possible to show that
this minimum formula is always well defined~\cite{Doumane17}. If
such minimum formula is a greatest fixed point operator, then the
thread is a $\tnu$-thread.
Note that a $\tnu$-thread is necessarily infinite.

Now we proceed by identifying threads along branches of proof
derivations. To this aim, we provide a precise definition of
\emph{branch}.

\begin{definition}[branch]
    \label{def:branch}
    A \emph{branch} of a proof derivation is a sequence $(\vdash \Sequent_0,
    \vdash \Sequent_1, \dots)$ of sequents such that $\vdash \Sequent_0$ occurs
    somewhere in the derivation and $\vdash \Sequent_{i+1}$ is a premise of the
    rule application that derives $\vdash \Sequent_i$ whenever $i+1$ is a valid
    index of the sequence.
\end{definition}

An infinite branch is valid if supported by a $\tnu$-thread that
originates somewhere therein.

\begin{definition}
    \label{def:valid-branch}
    Let $\gamma = (\vdash \Sequent_0, \vdash \Sequent_1, \dots)$ be an infinite
    branch in a derivation. We say that $\gamma$ is \emph{valid} if there exists
    $I \subseteq \Nat$ such that $(F_i)_{i\in I}$ is a $\tnu$-thread and $F_i
    \in \Sequent_i$ for every $i\in I$.
\end{definition}

\begin{definition}
    \label{def:valid-proof}
    A \muMALL derivation is \emph{valid} if so are its infinite branches.
\end{definition}

%%% Local Variables:
%%% mode: latex
%%% TeX-master: "main"
%%% End:

\section{Fair Termination of \CoreCSLL}
\label{sec:termination}

\newcommand{\cmap}{\sigma}
\newcommand{\amap}{\rho}
\newcommand{\extend}[2]{#1 \mapsto #2}
\newcommand{\even}[1]{\mathsf{even}(#1)}
\newcommand{\odd}[1]{\mathsf{odd}(#1)}

\newcommand{\dencode}[3]{\left\llbracket#1\right\rrbracket_{#2}^{#3}}
\newcommand{\cencode}[2]{\llbracket#1\rrbracket_{#2}}
\newcommand{\tencode}[1]{\llbracket#1\rrbracket}

In this section we prove that well-typed \CoreCSLL processes fairly
terminate.  We do so by appealing to the alternative
characterization of fair termination given by
\cref{thm:fair-termination}. Using that characterization and using
the fact that typing is preserved by reductions
(\cref{thm:preservation}), it suffices to show that well-typed
\CoreCSLL processes \emph{weakly} terminate.
To do that, we encode a well-typed \CoreCSLL process $P$ into a
(valid) \muMALL proof and we use the cut elimination property of
\muMALL to argue that $P$ has a finite maximal run.

\paragraph*{Encoding of types}

The encoding of \CoreCSLL types into \muMALL formulas is the map
$\tencode\cdot$ defined by
\begin{equation}
    \label{eq:encoding-type}
    \tencode{\tclient\Type} = \tmu\X.(\One \choice (\tencode\Type \tfork X))
    \qquad
    \tencode{\tserver\Type} = \tnu\X.(\Bot \branch (\tencode\Type \tjoin X))
\end{equation}
and extended homomorphically to all the other type constructors,
which are in one-to-one correspondence with the connectives and
constants of \muMALL. Notice that the image of the encoding is a
relatively small subset of \muMALL formulas in which different fixed
point operators are never intertwined.
Also notice that the encoding of the coexponentials does not follow
exactly their expansion in \cref{eq:coexp}. Basically, we choose to
interpret $\tclient\Type$ as a \emph{list} of clients rather than as
a \emph{tree} of clients, following to the intuition that clients
are queued when connecting to a server. The interpretation of
$\tserver$ follows as a consequence, as we want it to be the dual of
the interpretation of $\tclient$.
Note that this interpretation of the coexponential modalities is the
same used by Qian et al.~\cite{QianKavvosBirkedal21}.

\paragraph*{Encoding of typing contexts}

The next step is the encoding of \CoreCSLL contexts into \muMALL sequents.
Recall that a \muMALL sequent is a set of occurrences and that an occurrence is
a pair $\Formula_\address$ made of a formula $\Formula$ and an address
$\address$. In order to associate addresses with formulas, we parametrize the
encoding of \CoreCSLL contexts with an injective map $\cmap$ from \CoreCSLL
channels to addresses, since channels in (the domain of a) \CoreCSLL context
uniquely identify the occurrence of a type (and thus of a formula).
We write $\extend\x\address$ for the singleton map that associates
$x$ with the address $\address$ and $\cmap_1,\cmap_2$ for the union
of $\cmap_1$ and $\cmap_2$ when they have disjoint domains and
codomains.
Now, the encoding of a \CoreCSLL context is set of formulas defined by
\[
    \cencode{x_1 : \Type_1, \dots, x_n : \Type_n}{\cmap,\extend{x_1}{\address_1},\dots,\extend{x_n}{\address_n}} \eqdef
    \tencode{\Type_1}_{\address_1}, \dots, \tencode{\Type_n}_{\address_n}
\]

\paragraph*{Encoding of typing derivations}

Just like for the encoding of \CoreCSLL contexts, also the encoding of typing
derivations is parametrized by a map $\cmap$ from \CoreCSLL channels to
addresses. In addition, we also have to take into account the possibility that
\emph{restricted channels} are introduced in a \CoreCSLL context, which happens
in the rule \refrule\CutRule of \cref{tab:typing-rules}. The formula occurrence
corresponding to the type of this newly introduced channel must have an address
that is disjoint from that of any other occurrence. To guarantee this
disjointness, we parametrize the encoding of \CoreCSLL derivations by an
\emph{infinite stream} $\amap$ of pairwise distinct atomic addresses. Formally,
$\amap$ is an injective function $\Nat \to \AddressSet$.  We write $a\amap$,
$\even\amap$ and $\odd\amap$ for the streams defined by
\[
    (a\amap)(0) \eqdef 0
    \qquad
    (a\amap)(n + 1) \eqdef \amap(n)
    \qquad
    \even\amap(n) \eqdef \amap(2n)
    \qquad
    \odd\amap(n) \eqdef \amap(2n+1)
\]
respectively. In words, $a\amap$ is the stream of atomic addresses
that starts with $a$ and continues as $\amap$ whereas $\even\amap$
and $\odd\amap$ are the sub-streams of $\amap$ consisting of
addresses with an even (respectively, odd) index.

\newcommand{\Judgment}{\mathcal{J}}

The encoding of a \CoreCSLL typing derivation is coinductively defined by a map
$\dencode\cdot\cmap\amap$ which we describe using the following notation. For
every typing rule in \cref{tab:typing-rules}
\[
    \inferrule[\rulename{rule}]{\Judgment_1\\\cdots\\\Judgment_n}{\Judgment}
    \text{\qquad we write \qquad}
    \dencode{\inferrule{\Judgment_1\\\cdots\\\Judgment_n}\Judgment}\cmap\amap =
\pi
\]
meaning that $\pi$ is the \muMALL derivation resulting from the encoding of the
\CoreCSLL derivation for the judgment $\Judgment$ in which the last rule is an
application of $\rulename{rule}$. Within $\pi$ there will be instances of the
$\dencode{\Judgment_i}{\cmap_i}{\amap_i}$ for suitable $\cmap_i$ and $\amap_i$
standing for the encodings of the \CoreCSLL sub-derivations for the judgments
$\Judgment_i$ that we find as premises of $\rulename{rule}$.

There is a close correspondence between many \CoreCSLL typing rules and \muMALL
proof rules so we only detail a few interesting cases of the encoding,
%
% As a first example, consider the rules \refrule\CloseRule and
% \refrule\WaitRule which are handled as follows:
% \[
%   \dencode{
%     \inferrule{~}{
%       \wtp{\Close\x}{x : \One}
%     }
%   }{\cmap\extend\x\address}\amap =
%   \inferrule{~}{
%     \vdash \One_\address
%   }
%   \,\refrule\CloseRule
%   \qquad
%   \dencode{
%     \inferrule{
%       \wtp{P}\Context
%     }{
%       \wtp{\Wait\x.P}{\Context, x : \Bot}
%     }
%   }{\cmap\extend\x\address}\amap =
%   \inferrule{
%     \dencode{\wtp{P}\Context}\cmap\amap
%   }{
%     \vdash \cencode\Context\cmap, \Bot_\address
%   }
%   \,\refrule\WaitRule
% \]
%
% Note how the address associated with formulas obtained from the
% encoding of the type of $x$ is determined 
%
starting from the typing rules \refrule\ForkRule and \refrule\JoinRule. A
\muMALL typing derivation ending with an application of these rules is encoded
as follows:
\[
    \begin{array}{@{}r@{~}l@{}}
        \dencode{
            \inferrule{
                \wtp{P}{\ContextC, y : \TypeT}
                \\
                \wtp{Q}{\ContextD, x : \TypeS}            
            }{
                \wtp{\Fork\x\y{P}{Q}}{\ContextC, \ContextD, x : \TypeT \tfork \TypeS}
            }
        }{\cmap,\extend\x\address}\amap & =
        \inferrule{
            \dencode{\wtp{P}{\ContextC, y : \TypeT}}{\cmap,\extend\y{\address l}}{\even\amap}
            \\
            \dencode{\wtp{Q}{\ContextD, x : \TypeS}}{\cmap,\extend\x{\address r}}{\odd\amap}
        }{
            \vdash \cencode{\ContextC, \ContextD}\cmap, \tencode{\TypeT \tfork \TypeS}_\address
        }
        \,\refrule\ForkRule
        \\
        \dencode{
            \inferrule{
                \wtp{P}{\Context, y : \TypeT, x : \TypeS}
            }{
                \wtp{\Join\x\y.P}{\Context, x : \TypeT \tjoin \TypeS}
            }
        }{\cmap,\extend\x\address}\amap & =
        \inferrule{
            \dencode{\wtp{P}{\Context, y : \TypeT, x : \TypeS}}{\cmap,\extend\y{\address l},\extend\x{\address r}}\amap
        }{
            \vdash \cencode\Context\cmap, \tencode{\TypeT \tjoin \TypeS}_\address
        }
        \,\refrule\JoinRule
    \end{array}
\]

Notice that the types $\TypeT \tfork \TypeS$ and $\TypeT \tjoin \TypeS$
associated with $x$ in the conclusion of the rules are encoded into the
occurrences $\tencode{\TypeT \tfork \TypeS}_\address$ and $\tencode{\TypeT
\tjoin \TypeS}_\address$ where $\address$ is the address associated with $x$ in
$\cmap,\extend\x\address$. This address is suitably updated in the encoding of
the premises of the rules.
In the case of \refrule\ForkRule, the original stream $\amap$ of
atomic addresses is split into two disjoint streams in the encoding
of the premises to ensure that no atomic address is used twice.

% Most typing rules of \CoreCSLL except \refrule\CallRule and those for
% co-exponentials are in direct correspondence with proof rules of \muMALL, so in
% the following we will only detail a few interesting

Every application of \refrule\CallRule is simply erased in the encoding:
\[
    \dencode{
        \inferrule{
            \wtp{P}{\seqof{x : \Type}}
        }{
            \wtp{\Call\A{\seqof\x}}{\seqof{x : \Type}}
        }
    }\cmap\amap = \dencode{\wtp{P}{\seqof{x : \Type}}}\cmap\amap
\]

The validity of the \CoreCSLL typing derivation guarantees that there cannot be
an infinite chain of process invocations in a well-typed process. A proof of
this fact is given by \cref{lem:unfolded} in \cref{sec:extra-types}. For this
reason, the encoding of \CoreCSLL derivations is well defined despite the fact
that applications of \refrule\CallRule are erased.

% \[
%     \dencode{
%         \inferrule{
%             \wtp{P}{\Context, x : \TypeT}
%         }{
%             \wtp{\Select\x{\InTag_1}.P}{\Context, x : \TypeT \choice \TypeS}
%         }
%     }{\cmap\extend\x\address}\amap =
%     \inferrule{
%         \dencode{\wtp{P}{\Context, x : \TypeT}}{\cmap\extend\x{\address l}}\amap
%     }{
%         \vdash \cencode\Context\cmap, \tencode{\TypeT \choice \TypeS}_\address
%     }
%     \refrule\SelectRule
% \]

% \[
%     \dencode{
%         \inferrule{
%             \wtp{P}{\Context, x : \TypeT}
%             \\
%             \wtp{Q}{\Context, x : \TypeS}
%         }{
%             \wtp{\Case\x{P}{Q}}{\Context, x : \TypeT \branch \TypeS}
%         }
%     }{\cmap\extend\x\address}\amap =
%     \inferrule{
%         \dencode{\wtp{P}{\Context, x : \TypeT}}{\cmap\extend\x{\address l}}\amap
%         \\
%         \dencode{\wtp{Q}{\Context, x : \TypeS}}{\cmap\extend\x{\address r}}\amap
%     }{
%         \vdash \cencode\Context\cmap, \tencode{\TypeT \branch \TypeS}_\address
%     }
%     \,\refrule\CaseRule
% \]

Another case worth discussing is that of the rule \refrule\CutRule, which is
handled as follows:
\[
    \dencode{
        \inferrule{
            \wtp{P}{\ContextC, x : \Type}
            \\
            \wtp{Q}{\ContextD, x : \dual\Type}
        }{
            \wtp{\Cut\x{P}{Q}}{\ContextC, \ContextD}
        }
    }\cmap{a\amap} =
    \inferrule{
        \dencode{\wtp{P}{\ContextC, x : \Type}}{\cmap,\extend\x{a}}{\even\amap}
        \quad
        \dencode{\wtp{Q}{\ContextD, x : \dual\Type}}{\cmap,\extend\x{\dual{a}}}{\odd\amap}
    }{
        \vdash \cencode{\ContextC, \ContextD}\cmap
    }
    \,\refrule\CutRule
\]

The first address from the infinite stream $a\amap$, which is guaranteed to be
distinct from any other address used so far and that will be used in the rest of
the encoding, is associated with the newly introduced variable $x$. Similarly to
the case of \refrule\ForkRule, the tail of the stream is split in the encoding
of the two premises of \refrule\CutRule so as to preserve this guarantee.

We now consider the applications of \refrule\NilRule,
\refrule\ConsRule and \refrule\ServerRule which account for the most
relevant part of the encoding.  These rule applications are encoded
by considering the interpretation of the co-exponentials in terms of
least and greatest fixed points (\cref{eq:encoding-type}) and then
by applying the suitable \muMALL proof rules (\refrule\RecRule and
\refrule\CorecRule in particular).
We have
\[
    \dencode{
        \inferrule{~}{
            \wtp{\Nil\x}{x : \tclient\Type}
        }
    }{\cmap,\extend\x\address}\amap =
    \inferrule{
        \inferrule{
            \inferrule{~}{
                \vdash \One_{\address il}
            }
            \rlap{\refrule\CloseRule}
        }{
            \vdash \tencode{\One \choice (\Type \tfork \tclient\Type)}_{\address i}
        }
        \rlap{\refrule\SelectRule}
    }{
        \vdash \tencode{\tclient\Type}_\address
    }
    \refrule\RecRule
\]
for the applications of \refrule\NilRule and
\begin{multline*}
    \dencode{
        \inferrule{
            \wtp{P}{\ContextC, y : \Type}
            \\
            \wtp{Q}{\ContextD, x : \tclient\Type}
        }{
            \wtp{\Cons\x\y{P}Q}{\ContextC, \ContextD, x : \tclient\Type}
        }
    }{\cmap,\extend\x\address}\amap = \\
    \begin{prooftree}
        \[
            \[
                \dencode{
                    \wtp{P}{\ContextC, y : \Type}
                }{\cmap,\extend\y{\address irl}}{\even\amap}
                \qquad
                \dencode{
                    \wtp{Q}{\ContextD, x : \tclient\Type}
                }{\cmap,\extend\x{\address irr}}{\odd\amap}
                \justifies
                \vdash
                    \cencode{\ContextC, \ContextD}\cmap,
                    \tencode{\Type \tfork \tclient\Type}_{\address ir}
                \using\refrule\ForkRule
            \]
            \justifies
            \vdash
                \cencode{\ContextC, \ContextD}\cmap,
                \tencode{\One \choice (\Type \tfork \tclient\Type)}_{\address i}
            \using\refrule\SelectRule
        \]
        \justifies
        \vdash \cencode{\ContextC, \ContextD}\cmap, \tencode{\tclient\Type}_\address
        \using\refrule\RecRule
    \end{prooftree}
\end{multline*}
for the applications of \refrule\ConsRule. Finally, the applications
of \refrule\ServerRule are encoded thus:
\begin{multline*}
    \dencode{
        \inferrule{
            \wtp{P}{\Context, x : \tserver\Type, y : \Type}
            \\
            \wtp{Q}\Context
        }{
            \wtp{\Server\x\y{P}{Q}}{\Context, x : \tserver\Type}
        }
    }{\cmap,\extend\x\address}\amap = \\
    \begin{prooftree}
        \[
            \[
                \dencode{
                    \wtp{Q}\Context
                }\cmap\amap
                \justifies
                \vdash
                    \cencode\Context\cmap, \Bot_{\address il}
                \using\refrule\WaitRule
            \]
            \[
                \dencode{
                    \wtp{P}{\Context, x : \tserver\Type, y : \Type}
                }{\cmap,\extend\x{\address irr},\extend\y{\address irl}}\amap
                \justifies
                \vdash
                    \cencode\Context\cmap,
                    \tencode{\Type \tjoin \tserver\Type}_{\address ir}
                \using\refrule\JoinRule
            \]
            \justifies
            \vdash
                \cencode\Context\cmap,
                \tencode{\Bot \branch (\Type \tjoin \tserver\Type)}_{\address i}
            \using\refrule\CaseRule
        \]
        \justifies
        \vdash \cencode\Context\cmap, \tencode{\tserver\Type}_\address
        \using\refrule\CorecRule
    \end{prooftree}
\end{multline*}

Note that in this last case it is not necessary to split the stream $\amap$
since the $P$ and $Q$ branches of the server are mutually exclusive (the
reduction rules \refrule{r-connect} and \refrule{r-done} pick one or the other
branch, but not both). A similar thing happens in the encoding of the
applications of \refrule\CaseRule, not shown here.

\paragraph*{Validity of encoded typing derivations}

Now that we have shown how every \CoreCSLL typing derivation is encoded into a
\muMALL derivation, we argue that the encoding preserves validity. More
specifically, a valid \CoreCSLL typing derivation (\cref{def:valid-typing}) is
encoded into a valid \muMALL derivation (\cref{def:valid-proof}).
To see that this is the case, first observe that there is a one-to-one
correspondence between the infinite branches in the two derivations. From
\cref{def:valid-typing} we know that every infinite branch in a \CoreCSLL
derivation contains infinitely many applications of \refrule\ServerRule
concerning the same shared channel $x$ having type $\tserver\Type$ for some
$\Type$. In the encoded derivation, this translates to the existence of a
formula $\tencode{\tserver\Type}$ that occurs infinitely often in the sequents
making up this infinite branch. Now, suppose that the first occurrence of this
formula is associated with some address $\address$. From the encoding of
\refrule\ServerRule we can then build the thread
\[
    t \eqdef
    (
        \tencode{\tserver\Type}_\address,
        \tencode{\Bot \branch (\Type \tjoin \tserver\Type)}_{\address i},
        \tencode{\Type \tjoin \tserver\Type}_{\address ir},
        \tencode{\tserver\Type}_{\address irr},
        \dots
    )
\]
which is \emph{infinite}. Also note that $\InfOften{t} = \set{
\tencode{\tserver\Type}, \tencode{\Bot \branch (\Type \tjoin \tserver\Type)},
\tencode{\Type \tjoin \tserver\Type} }$, that $\minf\InfOften{t} =
\tencode{\tserver\Type}$, and that $\tencode{\tserver\Type}$ is a $\nu$-formula
by \cref{eq:encoding-type}.
In conclusion, $t$ is a $\tnu$-thread (\cref{def:nu-thread}) as required by the
validity condition for \muMALL pre-proofs (\cref{def:valid-proof}).

\paragraph*{Soundness of the type system}

Now that we know how to obtain a \muMALL proof from a well-typed \CoreCSLL
process we observe that each reduction rule of \DetCSLL corresponds to \emph{one
or more} principal reductions in a \muMALL proof~\cite[Figure 3.2]{Doumane17}.
In particular, the reductions \refrule{r-close}, \refrule{r-comm} and
\refrule{r-case} correspond to \emph{exactly one} principal reduction in \muMALL
(for $\One/\Bot$, ${\tfork}/{\tjoin}$ and ${\choice}/{\branch}$ respectively),
whereas \refrule{r-done} and \refrule{t-client} correspond to \emph{three}
subsequent principal reductions in \muMALL. For example, \refrule{r-connect}
corresponds to the principal reduction $\tmu$/$\tnu$ followed by
${\choice}/{\branch}$ followed by ${\tfork}/{\tjoin}$.
Using this correspondence between \DetCSLL and \muMALL, we can prove
that every \CoreCSLL process that is well typed in a context of the
form $x : \One$ is weakly terminating.

\begin{theorem}
  \label{thm:weak-termination}
  If $\wtp{P}{x : \One}$ then $P$ is weakly terminating.
\end{theorem}
\begin{proof}
    Let $a\amap$ be an infinite stream of pairwise distinct atomic addresses.
    Every reduction of $P$ can be mimicked by one or more principal reductions
    in the \muMALL proof $\dencode{\wtp{P}{x : \One}}{\extend\x{a}}\amap$.
    We know that \muMALL enjoys cut elimination \cite{Doumane17}. In particular,
    there cannot be an infinite sequence of principal reductions in a \muMALL
    proof \cite[Proposition 3.5]{Doumane17}.
    It follows that there is no infinite sequence of reductions starting from
    $P$ (using the \DetCSLL semantics), that is $P \dreds Q \ndred$ for some
    $Q$.
    From \cref{thm:preservation} we deduce $\wtp{Q}{x : \One}$ and from
    \cref{thm:df} we deduce $Q \pcong \Close\x$.
    We conclude $P \wred\pcong \Close\x \nred$. In other words, $P$ is weakly
    terminating.
\end{proof}

\begin{corollary}
  If $\wtp{P}{x : \One}$ then $P$ is fairly terminating.
\end{corollary}
\begin{proof}
  Straightforward consequence of
  \cref{thm:fair-termination,thm:weak-termination}.
\end{proof}

%%% Local Variables:
%%% mode: latex
%%% TeX-master: "main"
%%% End:

\newcommand{\True}{\mktag{true}}
\newcommand{\False}{\mktag{false}}
\newcommand{\If}[4]{\mkkeyword{if}~#1=#2~\mkkeyword{then}~#3~\mkkeyword{else}~#4}
\newcommand{\Drop}[1]{\mkkeyword{drop}~#1}
\newcommand{\CASClient}[2]{\textsf{Client}_{#1,#2}}
\newcommand{\CASServer}[1]{\textsf{CAS}_{#1}}
\newcommand{\CASClients}{\textsf{Clients}}

\section{Example: a Compare-and-Swap register}
\label{sec:example}

In this section we illustrate a more complex scenario of client-server
interaction that hightlights not only the fact that the server handles
connections sequentially in an unspecified order but also the fact that each
connection may change the server's internal state and affect other connections.
More specifically, we show a modeling of the Compare-and-Swap (CAS) register of
Qian et al.~\cite{QianKavvosBirkedal21} in \CoreCSLL. A CAS register holds a
boolean value $\True$ or $\False$ and is represented as a server that accepts
connections from clients. Each client sends two boolean values to the server, an
\emph{expected value} and a \emph{desired value}. If the expected value matches
the content of the register, then the register is overwritten with the desired
value. Otherwise, the register remains unchanged. We model boolean values as
choices made in some session $y$. For instance, we can model the sending of
$\True$ on $y$ by the selection $\Left\y$ and the sending of $\False$ on $y$ by
the selection $\Right\y$. In fact, in this section we write $\True$ and $\False$
as aliases for the labels $\LeftTag$ and $\RightTag$, respectively.

Below are two definitions for clients that differ for the expected and desired
values they send to the CAS register:
\[
    \Let{\CASClient\True\False}{y}{\Select\y\True.\Select\y\False.\Close\y}
    \qquad
    \Let{\CASClient\False\True}{y}{\Select\y\False.\Select\y\True.\Close\y}
\]

It is easy to see that both definitions are well typed. In particular, we can
derive $\wtp{\Call{\CASClient{b}{c}}{y}}{y : (\One \choice \One) \choice (\One
\choice \One)}$ for every $b,c\in\set{\True,\False}$ with two applications of
\refrule\SelectRule and one application of \refrule\CloseRule.
We combine two clients in a single pool as by the following definition
\[
    \Let\CASClients{x}{
        \Cons\x\y{\Call{\CASClient\True\False}\y}
        \Cons\x\y{\Call{\CASClient\False\True}\y}
        \Nil\x
    }
\]
for which we derive $\wtp{\Call\CASClients\x}{x :
\tclient((\One\choice\One)\choice(\One\choice\One))}$ using \refrule\ConsRule
and \refrule\NilRule.

For the CAS server we provide two definitions $\CASServer\True$ and
$\CASServer\False$ corresponding to the states in which the register holds the
value $\True$ and $\False$, respectively.
\begin{align*}
    \Let{\CASServer\True}{x,z&}{
        \begin{array}[t]{@{}l@{}l@{}}
        \Server\x\y{&
            \begin{array}[t]{@{}l@{}l@{}}
            \Case\y{&
                \Case\y{
                    \Wait\y.\Call{\CASServer\True}{x,z}
                }{
                    \Wait\y.\Call{\CASServer\False}{x,z}
                }
            }{\\&
                \Case\y{
                    \Wait\y.\Call{\CASServer\True}{x,z}
                }{
                    \Wait\y.\Call{\CASServer\True}{x,z}
                }
            }
        }{
            \end{array}
            \\&
            \Select\z\True.\Close\z
        }
        \end{array}
    }
    \\
    \Let{\CASServer\False}{x,z&}{
        \begin{array}[t]{@{}l@{}l@{}}
        \Server\x\y{&
            \begin{array}[t]{@{}l@{}l@{}}
            \Case\y{&
                \Case\y{
                    \Wait\y.\Call{\CASServer\False}{x,z}
                }{
                    \Wait\y.\Call{\CASServer\False}{x,z}
                }
            }{\\&
                \Case\y{
                    \Wait\y.\Call{\CASServer\True}{x,z}
                }{
                    \Wait\y.\Call{\CASServer\False}{x,z}
                }
            }
        }{
            \end{array}
            \\&
            \Select\z\False.\Close\z
        }
        \end{array}
    }
\end{align*}

The server in state $b\in\set{\True,\False}$ waits for connections on the shared
channel $x$. If there is no client, the server sends $b$ on $z$ and terminates.
If a client connects, then a session $y$ is established. At this stage the
server performs two input operations to receive the expected and desired values
from the client. If the expected value does \emph{not} match $b$, then the
desired value is ignored and the server recursively invokes itself in the same
state $b$. If the expected value matches $b$, then the server recursively
invokes itself in a state that matches the client's desired value.

It is not difficult to obtain derivations for the judgments
$\wtp{\Call{\CASServer{b}}{x,z}}{x :
\tserver((\Bot\branch\Bot)\branch(\Bot\branch\Bot)), z : \One \choice \One}$ for
every $b\in\set{\True,\False}$. These derivations are valid since every infinite
branch in them goes through an application of \refrule\ServerRule concerning the
channel $x$.
In conclusion, the CAS server is well typed and so is the composition
$\Cut\x{\Call\CASClients\x}{\Call{\CASServer\True}{x,z}}$.

Note that the process $\Cut\x{\Call\CASClients\x}{\Call{\CASServer\True}{x,z}}$
is not deterministic since it may reduce to either $\Select\z\True.\Close\z$ or
$\Select\z\False.\Close\z$ depending on the order in which clients connect.
Indeed, if $\CASClient\True\False$ connects first, then the state of the
register changes from $\True$ to $\False$ and then the connection with the
second client changes it back from $\False$ to $\True$. If, on the other hand,
$\CASClient\False\True$ connects first (because \refrule{s-pool-comm} is used),
then the initial state of the register does not change and then it is changed
from $\True$ to $\False$ when the client $\CASClient\True\False$ finally
connects.

%%% Local Variables:
%%% mode: latex
%%% TeX-master: "main"
%%% End:

\section{Concluding Remarks}
\label{sec:conclusion}

\CSLL~\cite{QianKavvosBirkedal21} is a non-deterministic session calculus based
on linear logic in which servers handle multiple client requests sequentially.
In this work we have targeted the problem of proving the termination of
well-typed \CSLL processes. To this aim, we have introduced \CoreCSLL, a variant
of \CSLL closely related to
\muMALL~\cite{BaeldeDoumaneSaurin16,Doumane17,BaeldeEtAl22}, the infinitary
proof system for multiplicative additive linear logic with fixed points. We have
shown that well-typed \CoreCSLL processes are fairly terminating by encoding
\CoreCSLL typing derivations into \muMALL proofs and using the cut elimination
property of \muMALL. Although fair termination is weaker than termination, it is
strong enough to imply livelock freedom, which was one of the motivations for
proving termination in the original \CSLL work~\cite{QianKavvosBirkedal21}. In
our work, fair termination is termination under the fairness assumption that
termination is not avoided forever (\cref{thm:fair-termination}). However,
inspection of our proof (\cref{sec:termination}) reveals that the fairness
assumption can be substantially weakened: the fair termination in \CoreCSLL is
reduced to the termination in \DetCSLL, meaning that fair termination in
\CoreCSLL is guaranteed if client requests are handled in order.

\CoreCSLL differs from the original \CSLL in a few ways.
In the interest of simplicity, we have chosen to omit constructs for
modeling (pools of) sequential clients and replicated servers which
are meant to be typed using the traditional exponential
modalities. These features are orthogonal to the ones we are
interested in and we think that they can be accommodated without
substantial challenges following the same technical development
illustrated in the present paper. In fact, the general support to
fixed points in \muMALL allows for this and other extensions, such
as (co)recursive session
types~\cite{LindleyMorris16,CicconePadovani22c}.
Another difference is that \CoreCSLL adopts a reduction semantics that is
completely ordinary for a process calculus. In particular, reductions are
\emph{not} allowed under prefixes, restrictions can\emph{not} be moved beyond
prefixes and (unrelated) prefixes can\emph{not} be swapped. Nonetheless, we are
able to relate the reduction semantics of \CoreCSLL with the cut reduction
strategy of \muMALL since \muMALL proofs, which can be infinite, are reduced
bottom-up. For this reason, we find that \muMALL provides a natural logical
foundation for session calculi based on linear logic.

Just like \CSLL, also \CoreCSLL is related to
\SILLS~\cite{BalzerPfenning17,BalzerToninhoPfenning19} and
\HCPND~\cite{KokkeMorrisWadler20}, two session calculi based on linear logic
that allow for races and non-determinism.
In \SILLS, sessions can be \emph{shared} among more than two communicating
processes. Access to a shared session is regulated by means of explicit
acquire/release actions that manifest themselves as special modalities in
session types. The flexibility gained by session sharing may compromise deadlock
freedom, which can be recovered by means of additional type
structure~\cite{BalzerToninhoPfenning19}.
\HCPND uses bounded exponentials~\cite{GirardScedrovScott92} to implement
client/server interactions in which the amount of channel sharing is known (and
bounded) in advance. Neither \CSLL nor \CoreCSLL require such bounds. For
example, the forwarder process $\Call{\ExpandLink{\tserver\Type}}{x,y}$ in
\cref{ex:forwarder} would be ill typed in \HCPND since the number of clients
that may be willing to connect on $x$ is not known \emph{a priori}.

%%% Local Variables:
%%% mode: latex
%%% TeX-master: "main"
%%% End:

\appendix

\section{Supplement to Section~\ref{sec:types}}
\label{sec:extra-types}

In the proofs of \cref{lem:pcong,lem:sr} below we only focus on the
\emph{derivability} of the typing judgment a structural
pre-congruence or a reduction, without worrying about the
\emph{validity} of the derivation. It is easy to see that validity
is preserved since both structural pre-congruence and reductions
either change a \emph{finite} region of the typing derivation or
\emph{remove} an entire sub-tree of the derivation (as in the case
of \refrule{r-case}). Either way, the fact that every infinite
branch in the residual derivation satisfies the validity conditions
(\cref{def:valid-typing}) follows from the hypothesis that the
initial typing derivation is valid.

\begin{lemma}
    \label{lem:pcong}
    If $\wtp{P}\Context$ and $P \pcong Q$ then $\wtp{Q}\Context$.
\end{lemma}
\begin{proof}
  By induction on the derivation of $P \pcong Q$ and by cases on the last rule
  applied. The proof is standard, we only discuss \refrule{s-par-pool} for
  illustration purposes.
  In this case $P = \Cut\z{\Cons\x\y{P_1}{P_2}}{P_3} \pcong
  \Cons\x\y{P_1}{\Cut\z{P_2}{P_3}} = Q$ where $z \not\in
  \fn{\Client\x\y.P_1}$.
  From \refrule\CutRule we deduce $\wtp{\Cons\x\y{P_1}P_2}{\Context_{12}, z :
  \TypeT}$ and $\wtp{P_3}{\Context_{3}, z : \dual\TypeT}$ where $\Context =
  \Context_{12}, \Context_3$.
  From \refrule\ConsRule and $z\not\in\fn{\Client\x\y.P_1}$ we deduce
  $\wtp{P_1}{\Context_1, y : \TypeS}$ and $\wtp{P_2}{\Context_2, x :
  \tclient\TypeS, z : \TypeT}$ where $\Context_{12} = \Context_1, \Context_2, x
  : \tclient\TypeS$.
  We derive $\wtp{\Cut\z{P_2}{P_3}}{\Context_2, \Context_3, x : \tclient\TypeS}$
  with one application of \refrule\CutRule.
  We conclude $\wtp{Q}\Context$ with one application of \refrule\ConsRule.
\end{proof}

\begin{lemma}
    \label{lem:sr}
    If $\wtp{P}\Context$ and $P \red Q$ then $\wtp{Q}\Context$.
\end{lemma}
\begin{proof}
  By induction on the derivation of $P \red Q$ and by cases on the
  last rule applied.

  \begin{itemize}
  \item \refrule{r-close} Then
    $P = \Cut\x{\Close\x}{\Wait\x.Q} \red Q$.  From
    \refrule\CutRule, \refrule\CloseRule and \refrule\WaitRule we
    deduce $\wtp{\Close\x}{x:\One}$ and
    $\wtp{\Wait\x.Q}{\Context,x:\Bot}$.
    From \refrule\WaitRule we conclude $\wtp{Q}\Context$.

  \item \refrule{r-comm} Then
    $P = \Cut\x{\Fork\x\y{P_1}{P_2}}{\Join\x\y.P_3} \red
    \Cut\y{P_1}{\Cut\x{P_2}{P_3}} = Q$.
    From \refrule\CutRule we deduce
    $\wtp{\Fork\x\y{P_1}{P_2}}{\Context_{12}, x : \TypeT \tfork
      \TypeS}$ and
    $\wtp{\Join\x\y.P_3}{\Context_3, x : \dual\TypeT \tjoin
      \dual\TypeS}$ where $\Context = \Context_{12}, \Context_3$.
    From \refrule\ForkRule we deduce
    $\wtp{P_1}{\Context_1, y : \TypeT}$ and
    $\wtp{P_2}{\Context_2, x : \TypeS}$ where
    $\Context_{12} = \Context_1, \Context_2$.
    From \refrule\JoinRule we deduce
    $\wtp{P_3}{\Context_3, y : \dual\TypeT, x : \dual\TypeS}$.
    We derive
    $\wtp{\Cut\x{P_2}{P_3}}{\Context_2,\Context_3,y:\dual\TypeT}$
    with one application of \refrule\CutRule.
    We conclude $\wtp{\Cut\y{P_1}{\Cut\x{P_2}{P_3}}}\Context$ with
    one application of \refrule\CutRule.

  \item \refrule{r-case} Then
    $P = \Cut\x{\Select\x{\InTag_i}.R}{\Case\x{Q_1}{Q_2}} \red
    \Cut\x{R}{Q_i} = Q$.
    From \refrule\CutRule, \refrule\SelectRule and \refrule\CaseRule
    we deduce
    $\wtp{\Select\x{\InTag_i}.R}{\Context_1,x:\Type_1
      \choice\Type_2}$ and
    $\wtp{\Case\x{Q_1}{Q_2}}{\Context_2,x:\dual{\Type_1}\branch\dual{\Type_2}}$
    where $\Context = \Context_1,\Context_2$.
    From \refrule\SelectRule we deduce
    $\wtp{R}{\Context_1,x:\Type_i}$.
    From \refrule\CaseRule we deduce
    $\wtp{Q_i}{\Context_2,x:\dual{\Type_i}}$ for $i=1,2$.
    We conclude $\wtp{\Cut\x{R}{Q_i}}{\Context}$ with one
    application of \refrule\CutRule.
    
    \item \refrule{r-connect} Then $P =
    \Cut\x{\Cons\x\y{P_1}P_2}{\Server\x\y{Q_1}{Q_2}} \red
    \Cut\y{P_1}{\Cut\x{P_2}{Q_1}} = Q$.
    From \refrule\CutRule, \refrule\ConsRule and \refrule\ServerRule we deduce
    $\wtp{\Cons\x\y{P_1}P_2}{\Context_{12}, x : \tclient\Type}$ and
    $\wtp{\Server\x\y{Q_1}{Q_2}}{\ContextD, x : \tserver\dual\Type}$ where
    $\Context = \Context_{12},\ContextD$.
    From \refrule\ConsRule we deduce $\wtp{P_1}{\Context_1,y:\Type}$ and
    $\wtp{P_2}{\Context_2,x:\tclient\Type}$ where $\Context_{12} =
    \Context_1,\Context_2$.
    From \refrule\ServerRule we deduce $\wtp{Q_1}{\ContextD, x :
    \tserver\dual\Type, y : \dual\Type}$.
    We derive $\wtp{\Cut\x{P_2}{Q_1}}{\Context_2, \ContextD, y : \dual\Type}$
    with one application of \refrule\CutRule.
    We conclude $\wtp{\Cut\y{P_1}{\Cut\x{P_2}{Q_1}}}\Context$ with another
    application of \refrule\CutRule.
    
  \item \refrule{r-done} Then $P = \Cut\x{\Nil\x}{\Server\x\y{R}{Q}} \red Q$.
    From \refrule\CutRule, \refrule\NilRule and \refrule\ServerRule we deduce
    $\wtp{\EmptyPool\x}{x:\tclient\Type}$ and $\wtp{\Server\x\y{R}{Q}}{\Context,
    x : \tserver\dual\Type}$.
    From \refrule\ServerRule we conclude $\wtp{Q}\Context$.

  \item \refrule{r-par} Then
    $P = \Cut\x{P_1}{P_2} \red \Cut\x{Q_1}{P_2} = Q$ where
    $P_1 \red Q_1$.
    From \refrule\CutRule we deduce $\wtp{P_1}{\Context_1,x:\Type}$ and
    $\wtp{P_2}{\Context_2,x:\dual\Type}$ where $\Context=\Context_1,\Context_2$.
    Using the induction hypothesis we derive $\wtp{Q_1}{\Context_1,x:\Type}$.
    We conclude $\wtp{\Cut\x{Q_1}{P_2}}\Context$ with an application of
    \refrule\CutRule.

  \item \refrule{r-pool} Then $P = \Cons\x\y{P_1}{P_2} \red \Cons\x\y{P_1}{Q_2}
    = Q$ where $P_1 \red Q_2$.
    From \refrule\ConsRule we deduce $\wtp{P_1}{\Context_1, y : \Type}$ and
    $\wtp{P_2}{\Context_2, x : \tclient\Type}$ where $\Context = \Context_1,
    \Context_2, x : \tclient\Type$.
    Using the induction hypothesis we derive $\wtp{Q_2}{\Context_2, x :
    \tclient\Type}$.
    We conclude $\wtp{\Cons\x\y{P_1}{Q_2}}\Context$ with an application of
    \refrule\ConsRule.

  \item \refrule{r-struct} Using the induction hypothesis with two applications
  of \cref{lem:pcong}.
    \qedhere
  \end{itemize}
\end{proof}

In order to prove deadlock freedom it is convenient to introduce \emph{reduction
contexts} to make it easy to refer to unguarded sub-terms of a process. A
reduction context is basically a process with a single hole denoted by $\Hole$. 
\[
  \textbf{Reduction context}
  \qquad
  \RC, \RD ~~::=~~ \Hole ~~\mid~~ \Cut\x\RC{P} ~~\mid~~ \Cut\x{P}\RC
\]

Note that holes cannot occur in the tail of client pools, that is
$\Cons\x\y{P}\RC$ is \emph{not} a reduction context even though the tail of a
client pool may reduce by means of \refrule{r-pool}. The point is that, in order
to prove deadlock freedom, it is never necessary to reduce the tail of a client
pool.
Hereafter we write $\RC[P]$ for the process obtained by replacing the hole in
$\RC$ with $P$. Note that this notion of replacement may capture some channels
occurring free in $P$.

Before addressing deadlock freedom, we prove the following proximity lemma,
showing that it is always possible to move a restriction close to a process in
which the restricted channel occurs free.

\begin{lemma}
    \label{lem:proximity}
    If $x\in\fn{P}\setminus(\fn\RC\cup\bn\RC)$ then $\Cut\x{\RC[P]}Q \pcong
    \RD[\Cut\x{P}{Q}]$ for some $\RD$.
\end{lemma}
\begin{proof}
    By induction on $\RC$ and by cases on its shape. We do not detail symmetric
    cases and we assume, without loss of generality, that $\fn{Q} \cap \bn\RC =
    \emptyset$.

    \begin{itemize}
        \item $\RC = \Hole$. We conclude by taking $\RD \eqdef \Hole$ and by
        reflexivity of $\pcong$.

        \item $\RC = \Cut\y{\RC'}{R}$. Then $x \in \fn{P} \setminus (\fn{\RC'}
        \cup \bn{\RC'} \cup \fn{R} \cup \set\y)$. We derive
        \[
            \begin{array}[b]{r@{~}c@{~}ll}
                \Cut\x{\RC[P]}{Q} & = & \Cut\x{\Cut\y{\RC'[P]}R}{Q}
                & \text{by definition of $\RC$}
                \\
                & \pcong & \Cut\x{Q}{\Cut\y{\RC'[P]}{R}}
                & \text{by \refrule{s-par-comm}}
                \\
                & \pcong & \Cut\y{\Cut\x{Q}{\RC'[P]}}{R}
                & \text{by \refrule{s-par-assoc} since $x\not\in\fn{R}$, $y\not\in\fn{Q}$}
                \\
                & \pcong & \Cut\y{\Cut\x{\RC'[P]}Q}{R}
                & \text{by \refrule{s-par-comm}}
                \\
                & \pcong & \Cut\y{\RD'[\Cut\x{P}{Q}]}{R}
                & \text{by ind. hyp. for some $\RD'$}
                \\
                & = & \RD[\Cut\x{P}{Q}]
                & \text{by taking $\RD \eqdef \Cut\y{\RD'}R$}
            \end{array}
            \qedhere
          \]
        \end{itemize}
\end{proof}

\newcommand{\tdepth}[1]{\mathsf{cd}(#1)}

The next auxiliary result proves that, in a well-typed process, a finite number
of applications of \refrule{s-call} is always sufficient to unfold all of the
process invocations occurring in it. To this aim, we introduce some more
terminology on processes.
We say that $P$ is a \emph{guard} if it is not a parallel composition or a
process invocation. Note that every guard specifies a topmost action on some
channel $x$. In this case, we say that $P$ is an $x$-guard. We say that $P$ is
\emph{unguarded} in $Q$ if $Q = \RC[P]$ for some $\RC$.
We say that $P$ is \emph{unfolded} if $P = \RC[Q]$ implies that $Q$ is not an
invocation.

\begin{lemma}
    \label{lem:unfolded}
    If $\wtp{P}\Context$ then there exists an unfolded $Q$ such that $P \pcong
    Q$.
\end{lemma}
\begin{proof}
  Let the \emph{call depth} of $P$ be the natural number $\tdepth{P}$
  inductively defined as follows:
  \[
    \tdepth{P} =
    \begin{cases}
      1 + \tdepth{Q} & \text{if $P = \Call\A{\seqof\x}$ and $\Let\A{\seqof\x}{Q}$} \\
      1 + \max\set{\tdepth{P_1},\tdepth{P_2}} & \text{if
        $P = \Cut\x{P_1}{P_2}$} \\
      0 & \text{otherwise}
    \end{cases}
  \]

  Roughly, $\tdepth{P}$ is the maximum depth in the typing derivation of $P$
  where an unguarded guard is encountered.
  To see that $\tdepth{P}$ is well defined, recall that in every infinite branch
  of a valid typing derivation there are infinitely many applications of
  \refrule\ServerRule and that a process of the form $\Server\x\y{Q}{R}$ is a
  guard.
  Therefore, the value of $\tdepth{P}$ is only determined by the portion of
  $P$'s derivation tree that stops at each occurrence of a guard. This portion
  is finite.
  The proof proceeds by induction on $\tdepth{P}$ and by cases on the shape of
  $P$. The desired $Q$ is obtained by applying \refrule{s-call} each time an
  unguarded invocation is encountered and the induction guarantees that this
  rewriting is finite.
\end{proof}

\thmdf*
\begin{proof}
    By \cref{lem:unfolded} we may assume, without loss of generality, that $P$
    is unfolded. We want to show that there are two $x$-guards in $P$ that can
    synchronize. To this aim, let $\threads{P}$ be inductively defined as
    \[
        \threads{P} =
        \begin{cases}
            \threads{P_1} + \threads{P_2} & \text{if $P = \Cut\x{P_1}{P_2}$} \\
            1 & \text{otherwise}
        \end{cases}
    \]
    and let $\channels{P}$ be inductively defined as
    \[
        \channels{P} =
        \begin{cases}
            1 + \channels{P_1} + \channels{P_2} & \text{if $P = \Cut\x{P_1}{P_2}$} \\
            0 & \text{otherwise}
        \end{cases}
    \]

    In words, $\threads{P}$ counts the number of unguarded guards in $P$ whereas
    $\channels{P}$ counts the number of unguarded restrictions in $P$.
    It is easy to prove that $\threads{P} > \channels{P}$. So, there must be at
    least one channel name $x$ such that $P$ contains two unguarded $x$-guards.
    That is, $P = \RC[\Cut\x{\RC_1[P_1]}{\RC_2[P_2]}]$ and both $P_1$ and $P_2$
    are $x$-guards and $x\not\in\fn{\RC_i}\cup\bn{\RC_i}$ for $i=1,2$.
    Then we derive
    \[
        \begin{array}{rcll}
            P & = & \RC[\Cut\x{\RC_1[P_1]}{\RC_2[P_2]}]
            & \text{by definition of $P$, $P_1$ and $P_2$}
            \\
            & \pcong & \RC[\RD_1[\Cut\x{P_1}{\RC_2[P_2]}]]
            & \text{by \cref{lem:proximity} for some $\RD_1$}
            \\
            & \pcong & \RC[\RD_1[\Cut\x{\RC_2[P_2]}{P_1}]]
            & \text{by \refrule{s-par-comm}}
            \\
            & \pcong & \RC[\RD_1[\RD_2[\Cut\x{P_2}{P_1}]]]
            & \text{by \cref{lem:proximity} for some $\RD_2$}
        \end{array}
    \]

    Now we reason by cases on the shape of $P_1$ and $P_2$, knowing that they
    are $x$-guards and that they are well typed in contexts that contain the
    associations $x : \Type$ and $x : \dual\Type$ for some $\Type$.
    If $P_2 = \Client\x\y.Q$ then $P_1 = \Server\x\y{Q_1}{Q_2}$ and $P$ may
    reduce using \refrule{r-connect}. The case in which $P_1 = \Client\x\y.Q$ is
    symmetric and can be handled in a similar way with an additional application
    of \refrule{s-par-comm}.
    The cases in which one of $P_1$ and $P_2$ is $\Nil\x$ can be handled
    analogously, deducing that $P$ may reduce using \refrule{r-done}. The only
    cases left are when neither $P_1$ nor $P_2$ is a client or $\Nil\x$.
    Then, $P_1$ and $P_2$ must be $x$-guards beginning with dual actions which
    can synchronize using one of the rules \refrule{r-close}, \refrule{r-comm}
    or \refrule{r-case}, possibly with the help of an application of
    \refrule{s-par-comm}.
\end{proof}

%%% Local Variables:
%%% mode: latex
%%% TeX-master: "main"
%%% End:

\bibliography{main}

\end{document}